\documentclass[twoside,11pt]{article}
\usepackage{jair, theapa, rawfonts}

\usepackage{url}
\usepackage[utf8]{inputenc}
\usepackage[small]{caption}
\usepackage{graphicx}
\usepackage{xcolor}
\usepackage{amsmath}
\usepackage{amsthm}
\usepackage{amssymb}
\usepackage{booktabs}
\usepackage{algorithm} 
\usepackage[noend]{algorithmic} 
\usepackage[shortlabels]{enumitem}

\usepackage{cleveref}
\usepackage{tikz}
\usetikzlibrary{positioning,shapes.misc}

\usepackage{calc} 
\usepackage{booktabs} 
\usepackage{multirow} 

\newtheorem{theorem}{Theorem}[section]
\newtheorem{corollary}[theorem]{Corollary}
\newtheorem{proposition}[theorem]{Proposition}

\theoremstyle{definition}
\newtheorem{definition}[theorem]{Definition}

\newcommand{\citet}[1]{\citeauthor{#1}~\citeyear{#1}}
\newcommand{\citep}{\cite}

\usepackage{bbm}

\allowdisplaybreaks

\newcommand{\vv}{\mathbf{v}}

\newcommand{\SW}{\text{SW}}

\newcommand{\blue}{\text{\em blue}}
\newcommand{\red}{\text{\em red}}
\newcommand{\ept}{\text{\em empty}}
\newcommand{\true}{\text{\em true}}

\newcommand{\tree}{\text{\em tree}}

\ShortHeadings{Welfare Guarantees in Schelling Segregation}
{Bullinger, Suksompong, \& Voudouris}
\firstpageno{1}

\begin{document}

\title{Welfare Guarantees in Schelling Segregation\thanks{This work has been supported by the Deutsche Forschungsgemeinschaft under grant BR 2312/12-1, by the European Research Council (ERC) under grant number 639945 (ACCORD), and by an NUS Start-up Grant.
A preliminary version of the paper appeared in Proceedings
of the 35th AAAI Conference on Artificial Intelligence \citep{BullingerSuVo21}. 
This version contains a new section on computing optimal assignments (\Cref{sec:compopt}), additional results (\ref{thm:PO-welfare-tree}, \ref{cor:PO-price-tree}, \ref{prop:b/r-constant}, \ref{cor:b/r-constant}, \ref{prop:util-egal}, \ref{thm:UVO-welfare-1}, \ref{thm:deg-2-positive-algo}, \ref{prop:egal-zeropos}, \ref{prop:egal-highlow}), as well as all proofs omitted from the conference version.}}

\author{\name Martin Bullinger \email martin.bullinger@in.tum.de \\
       \addr Technische Universit\"{a}t München, Germany
       \AND
       \name Warut Suksompong \email warut@comp.nus.edu.sg \\
       \addr National University of Singapore, Singapore
       \AND
       \name Alexandros A. Voudouris \email alexandros.voudouris@essex.ac.uk \\
       \addr University of Essex, United Kingdom}

\maketitle

\begin{abstract}
Schelling's model is an influential model that reveals how individual perceptions and incentives can lead to residential segregation. Inspired by a recent stream of work, we study welfare guarantees and complexity in this model with respect to several welfare measures. First, we show that while maximizing the social welfare is NP-hard, computing an assignment of agents to the nodes of any topology graph with approximately half of the maximum welfare can be done in polynomial time. We then consider Pareto optimality, introduce two new optimality notions based on it, and establish mostly tight bounds on the worst-case welfare loss for assignments satisfying these notions as well as the complexity of computing such assignments. In addition, we show that for tree topologies, it is possible to decide whether there exists an assignment that gives every agent a positive utility in polynomial time; moreover, when every node in the topology has degree at least $2$, such an assignment always exists and can be found efficiently.
\end{abstract}

\section{Introduction}\label{sec:intro}
Schelling's model was proposed half a century ago to illustrate how individual perceptions and incentives can lead to racial segregation, and has been used to study this phenomenon in residential metropolitan areas in particular \citep{Schelling69,Schelling71}. 
The model is rather simple to describe. 
There are a number of agents, each of whom belongs to one of two predetermined types and occupies a location; in his original work, Schelling assumed that the locations are cells of a rectangular board, which can be represented as a grid graph.
Every agent would like to occupy a node on the graph such that the fraction of other agents of the same type in the neighborhood of that node is at least a predefined tolerance threshold $\tau \in [0,1]$.
If this condition is not met for an agent, then the agent can relocate to a randomly chosen empty node on the grid.
One of the most surprising findings of Schelling is that, starting from a random initial assignment of the agents to the nodes of the grid, the dynamics may converge to segregated assignments even when $\tau \approx 1/3$, contrasting the intuition that segregation should happen only when $\tau \geq 1/2$.

Throughout the years, hundreds of researchers in sociology and economics reconfirmed Schelling's observations and made similar ones for numerous variants of the model using computer simulations---see, e.g.,~\citep{ClarkFo08}. 
More recent work, mainly in computer science, performed rigorous analyses of such variants, some of which are quite close to the original model, and showed that the dynamics according to which the agents relocate converges to assignments in which the agents form large monochromatic regions (that is, subgraphs consisting only of agents of the same type); in addition, this line of work established bounds on the size of these regions.
We refer to the papers~\citep{Pollicott2001dynamics,Young01,Zhang04,Pancs2007proximity,BrandtImKa12,BarmpaliasElLe14,BarmpaliasElLe16,BhaktaMiRa14,ImmorlicaKlLu17} for results of this flavor.

While most of the literature on Schelling's model has focused on properties related to segregation between the two types, segregation itself is only one side of the story, especially when we allow different, possibly more complex location graphs. 
Given that the agents are willing to relocate to be close to other agents of the same type, another natural question is whether the resulting assignments satisfy some sort of {\em efficiency}. 
This has been considered in part by a recent array of papers~\citep{ChauhanLeMo18,EchzellFrLe19,ElkindGaIg19,AgarwalElGa20,BiloBiLe20,ChanIrTh20,KanellopoulosKyVo20}, which have studied Schelling's model from a game-theoretic perspective. 
In particular, instead of randomly relocating, the agents are assumed to be strategic and each of them aims to select a location that maximizes her {\em utility}, defined as the fraction of same-type agents in her neighborhood. 

Besides questions related to the existence and computation of equilibria (i.e., assignments in which no agent has an incentive to relocate in order to increase her utility), the authors of some of the aforementioned papers have also studied the efficiency of assignments in terms of {\em social welfare}, defined as the total utility of the agents.
For this objective, these authors have shown that computing assignments (not necessarily equilibria) maximizing the social welfare is NP-hard under specific assumptions about the graph and the behavior of the agents. 
Furthermore, they established several bounds on the worst-case ratio between the maximum social welfare (achieved by any possible assignment) and the social welfare of the best or worst equilbrium assignment, also known as the {\em price of stability}~\citep{AnchelevichDaKl08} and the {\em price of anarchy}~\citep{KP99}, respectively.
These ratios quantify the welfare that is lost due to the agents aiming to maximize their individual utilities rather than their collective welfare.

Inspired by this active stream of work, we study welfare guarantees and complexity in Schelling's model, not only with respect to the social welfare, but also to different notions of efficiency, such as Pareto optimality and natural variants of it. 

\subsection{Our Contribution}
Our setting consists of $n$ agents partitioned into two types, and a location graph known as the {\em topology}; agents of the same type are ``friends'', and agents of different types are ``enemies''. 
Each agent is assigned to a single node of the graph, and the utility of the agent is defined as the fraction of her friends among the agents in her neighborhood. 

We start by considering the social welfare.
We show that for any topology and any distribution of the agents into types, there always exists an assignment with social welfare at least $n/2-1$, and we provide a polynomial-time algorithm for computing such an assignment.
Since the social welfare never exceeds $n$, our algorithm produces an assignment with at least approximately half of the maximum social welfare. 
We complement this result by showing that maximizing the social welfare is NP-hard, even when the topology is a graph such that the number of nodes is equal to the number of agents. This improves upon previous hardness results of \citet{ElkindGaIg19} and \citet{AgarwalElGa20} whose reductions use instances with ``stubborn agents'' (who are assigned to fixed nodes in advance and cannot move), and either a topology with the number of nodes larger than the number of agents, or at least three types of agents instead of just two. These results are presented in Section~\ref{sec:welfare}.

\begin{table}[tb]
\centering
\begin{tabular}{lp{\widthof{$\Omega(n)$~[Prop. \ref{prop:star-bound}]}}p{\widthof{$\mathcal O(n\sqrt n)$~[Thm.~\ref{thm:PO-welfare}]}}p{\widthof{Welfare guarantee}}}
\toprule
\multirow{2}{*}{Welfare notion $P$} &
      \multicolumn{2}{c}{Price of $P$} & \multirow{2}{*}{Welfare guarantee}\\
    & Lower bound & Upper bound \\
\midrule
Maximum welfare& $1$ & $1$ & $\frac n2-1$~[Thm.~\ref{thm:randomized}]\\
GWO & $\Omega(n)$~[Prop.~\ref{prop:star-bound}] & $ O(n)$~[Thm.~\ref{thm:gwo-welfare}] & $\frac n{n-1}$~[Thm.~\ref{thm:gwo-welfare}] \\
UVO & $\Omega(n)$~[Thm.~\ref{thm:uvo-price}] & $ O(n)$~[Thm.~\ref{thm:uvo-price}] & $1$~[Thm.~\ref{thm:UVO-welfare-1}] \\
\midrule
PO ~ \parbox[c]{\widthof{approx. balanced}}{general\\trees \\ approx. balanced} & \parbox[c]{\widthof{approx. balanced}}{ $\Omega(n)$~[Prop.~\ref{prop:star-bound}] \\ $\Omega(n)$~[Cor.~\ref{cor:PO-price-tree}] \\  $\Omega(n)$~[Cor.~\ref{cor:b/r-constant}] } & \parbox[c]{\widthof{approx. balanced}}{ $ O(n\sqrt n)$~[Thm.~\ref{thm:PO-welfare}] \\ $ O(n)$~[Cor.~\ref{cor:PO-price-tree}]  \\  $ O(n)$~[Cor.~\ref{cor:b/r-constant}] } & \parbox[c]{\widthof{approx. balanced}}{ $\frac 1 {\sqrt n}$~[Thm.~\ref{thm:PO-welfare}] \\ $\frac n{n-1}$~[Thm.~\ref{thm:PO-welfare-tree}] \\  $\Omega(1)$~[Prop.~\ref{prop:b/r-constant}] } \\
\bottomrule
\end{tabular}
\caption{
An overview of our results on the price and welfare guarantees of the optimality notions we consider. The `approximately balanced' results for Pareto optimality hold when the numbers of agents of the two types are within a constant factor of each other (and the number of agents is equal to the number of nodes).
Combining the lower and upper bounds, we obtain a price of $\Theta(n)$ for all welfare notions except for maximum welfare (whose price is trivially $1$) and  PO on general topologies.
}
\label{tab:results}
\end{table}

Even if an assignment does not maximize the social welfare, it can still be optimal in other senses. 
With this in mind, in Section~\ref{sec:notions}, we turn our attention to different notions of optimality. 
In particular, we consider the well-known notion of {\em Pareto optimality} (PO), according to which it should not be possible to improve the utility of an agent without decreasing that of another agent. 
We also introduce two variants of PO, called {\em utility-vector optimality} (UVO) and {\em group-welfare optimality} (GWO), which are particularly appropriate for Schelling's model and may be of interest in other settings as well.
Informally, an assignment is UVO if we cannot improve the sorted utility vector of the agents, and GWO if it is not possible to increase the total utility of one type of agents without decreasing that of the other type. 
We prove several results on these three notions of optimality. 
First, while UVO and GWO imply PO by definition, we show that they are not implied by each other or by PO. Then, for each $P \in \{\text{PO},\text{UVO},\text{GWO}\}$, we establish mostly tight bounds on the {\em price of $P$}, which is an analogue of the price of anarchy: the price of $P$ is defined as the worst-case ratio between the maximum social welfare (among all assignments) and the minimum social welfare among all assignments satisfying $P$.\footnote{Note that an analogue of the price of \emph{stability}, where we consider the worst-case ratio between the maximum social welfare and the \emph{maximum} social welfare among assignments satisfying the optimality notion, is uninteresting: for all of the optimality notions we consider, this price is simply $1$.} 
Several of our results in Sections~\ref{sec:welfare} and \ref{sec:notions} are summarized in Table~\ref{tab:results}.

Next, in \Cref{sec:compopt}, we address the complexity of computing assignments satisfying different optimality notions.
Our NP-hardness reduction for social welfare maximization in \Cref{sec:welfare} also yields a corresponding hardness for GWO.
We then consider \emph{perfect} assignments, in which every agent receives the maximum utility of $1$, and show that deciding whether such an assignment exists is NP-complete.
As consequences, we obtain hardness results for computing a UVO or PO assignment, as well as for maximizing the \emph{egalitarian} welfare (defined as the minimum utility among all agents) and the \emph{Nash} welfare (defined as the product of the agents' utilities).
When perfect assignments do not exist, a reasonable relaxation is to require that every agent receives the maximum utility that she can receive in any assignment for that instance; we call assignments satisfying this requirement \emph{individually optimal}.
While deciding whether an individually optimal assignment exists is again NP-complete in general, when the number of agents is equal to the number of nodes in the topology, we present a characterization of instances admitting such an assignment---this characterization allows us to solve the decision problem in polynomial time for the special case.

Finally, another important measure of efficiency is the number of agents who receive a positive utility in the assignment. 
Even though only requiring the utility to be nonzero seems minimal, there exist simple instances in which not all of the agents can obtain a positive utility simultaneously. 
We show that for trees, it is possible to decide in polynomial time whether there exists an assignment such that all agents receive a positive utility. 
We then observe that it is always possible to guarantee a positive utility for at least half of the agents.
Moreover, when every node in the topology has degree at least $2$, an assignment in which all agents receive a positive utility is guaranteed to exist, and such an assignment can be computed in polynomial time. 
These results are presented in Section~\ref{sec:positive}.  

\subsection{Further Related Work}
As already mentioned, Schelling's model and its variants have been studied extensively from many different perspectives in several disciplines. 
For an overview of early work on the model, we refer the reader to \citep{ImmorlicaKlLu17}. 

Most related to our present work are the papers \citep{ElkindGaIg19,AgarwalElGa20,BiloBiLe20,KanellopoulosKyVo20}, which studied game-theoretic and complexity questions related to the social welfare in Schelling games. In particular, \citet{ElkindGaIg19} considered \emph{jump Schelling games} in which there are $k \geq 2$ types of agents, and the topology is a graph with more nodes than agents so that there are empty nodes to which unhappy agents can jump. They showed that equilibrium assignments do not always exist, proved that computing equilibrium assignments and assignments with social welfare close to $n$ (the maximum possible) is NP-hard, and bounded the price of anarchy and stability for both general and restricted games. 

Later on, \citet{AgarwalElGa20} considered the complement case of \emph{swap Schelling games} in which the number of nodes in the topology is equal to the number of agents; since there are no empty nodes to which the agents can jump, the agents can increase their utility only by swapping positions pairwise. For this setting, the authors showed results similar to those of \citet{ElkindGaIg19}.
\citet{BiloBiLe20} improved some of the price of anarchy bounds of \citet{AgarwalElGa20}, and also studied a variation of the model in which the agents have a restricted view of the topology and can only swap with their neighbors. 
\citet{KanellopoulosKyVo20} investigated the price of anarchy and stability in jump Schelling games, but with a slightly different utility function according to which an agent considers herself as part of her set of neighbors.

Schelling games are closely related to variants of hedonic games, most notably unweighted \emph{fractional hedonic games}~\citep{Aziz19fractional}, in which there is a set of agents and an unweighted graph that indicates friendship relations among them. In such games, the agents split into disjoint coalitions and, similarly to Schelling games, the utility of each agent is equal to the fraction of her friends in her coalition. The main difference between the two models is that, in Schelling games, the agents occupy the nodes of a topology graph, thereby leading to overlapping coalitions. 

The price of Pareto optimality was first considered in the context of fractional hedonic games by \citet{ElkindFaFl20}, and was also implicitly studied by \citet{Bullinger20}. Since Pareto optimality is a fundamental notion in various settings, its price has also been studied in the context of social distance games~\citep{BalliuFlOl17} and fair division \citep{BeiLuMa19}. To the best of our knowledge, this is the first time that Pareto optimality is studied in Schelling's model.

\section{Preliminaries}
Let $N=\{1, \dots, n\}$ be a set of $n \geq 2$ {\em agents}.
The agents are partitioned
into two different {\em types} (or {\em colors}), red and blue.
Denote by $r$ and $b$ the number of red and blue agents, respectively; we have $r+b = n$.
The distribution of agents into types is called \emph{balanced} if $|r-b|\le 1$.
We say that two agents $i, j\in N$ such that $i\neq j$ are {\em friends} if $i$ and $j$ are of the same type; otherwise we say that they are {\em enemies}. 
For each $i\in N$, we denote the set of all friends of agent $i$ by $F(i)$.

A {\em topology} is a simple connected undirected graph $G=(V,E)$, where $V=\{v_1,\dots,v_t\}$.
Each agent in $N$ has to select a node of this graph so that there are no collisions.
A tuple $I=(N,G)$ is called a {\em Schelling instance}.
Given a set of agents $N$  and a topology $G=(V, E)$ with $|V|\ge n$, an {\em assignment} is an $n$-tuple
$\vv=(v(1),\dots, v(n))\in V^n$ such that $v(i)\neq v(j)$ for all $i, j\in N$ with $i\neq j$;
here, $v(i)$ is the node of the topology where agent~$i$ is positioned.
A node $v \in V$ is {\em occupied} by agent $i$ if $v = v(i)$.
For a given assignment~$\vv$ and an agent $i\in N$,
let $N_i(\vv)=\{j \in N: \{v(i),v(j)\} \in E \}$ be the set of neighbors of
agent~$i$. Let $f_i(\vv)=|N_i(\vv) \cap F(i)|$ be the number of neighbors of $i$
in $\vv$ who are her friends. Similarly,
let $e_i(\vv) = |N_i(\vv)| - f_i(\vv)$ be the number
of neighbors of $i$ in $\vv$ who are her enemies.
Following prior work, 
we define the utility $u_i(\vv)$ of an agent $i$ in $\vv$ to be $0$ if $|N_i(\vv)|=0$; otherwise, her utility is defined as the fraction of her friends among the agents in her neighborhood:
$$
u_i(\vv) = \frac{f_i(\vv)}{|N_i(\vv)|} = \frac{f_i(\vv)}{f_i(\vv)+e_i(\vv)}.
$$
The {\em social welfare} of an assignment $\vv$ is defined as the total utility
of all agents:
$$
\SW(\vv) = \sum_{i \in N} u_i(\vv).
$$
Let $\vv^*(I)$ be an assignment that maximizes the social welfare for a given instance $I$; we refer to it as a {\em maximum-welfare} assignment.
Note that for any assignment $\vv$, we have $u_i(\vv) \le 1$, and so $\SW(\vv^*) \le n$.
Denote by $\SW_R(\vv)$ and $\SW_B(\vv)$ the sum of the utilities of the red and blue agents, respectively; we have $\SW_R(\vv) + \SW_B(\vv) = \SW(\vv)$.

\section{Social Welfare}\label{sec:welfare}

The first question we address is whether a high social welfare can always be achieved in any Schelling instance.
Even though it may seem that we can obtain high welfare simply by grouping the agents of each type together, given the possibly complex topology in combination with the distribution of agents into types, it is unclear how this idea can be executed in general or what guarantee it
results in.
Nevertheless, we show that high welfare is indeed always achievable.
Moreover, we provide a tight lower bound on the maximum welfare for each number of agents.

For any positive integer $n$, define
\[
g(n) = 
\begin{cases}
\frac{n(n-2)}{2(n-1)} \qquad &\text{ if } n \text{ is even};\\
\frac{n-1}{2} \qquad &\text{ if } n \text{ is odd}.
\end{cases}
\]
Note that $g(n) \ge n/2 - 1$ for all $n$.
Our approach is to choose an assignment uniformly at random among all possible assignments.
Equivalently, we place agents in the following iterative manner: for an arbitrary unoccupied node, assign a uniformly random agent who is unassigned thus far.
We show that the expected welfare of the assignment resulting from this simple randomized algorithm is at least $g(n)$, which implies the existence of an assignment with this welfare guarantee.

\begin{theorem}
\label{thm:randomized}
For any Schelling instance with $n$ agents, there exists an assignment with social welfare at least $g(n)$.
Moreover, the bound $g(n)$ cannot be improved.
\end{theorem}

\begin{proof}
First, note that we may assume that the number of agents is equal to the number of nodes by restricting our attention to an arbitrary connected subgraph of $G$ with the desired size. 
For $v_i\in V$, let 
$
N_{v_i} = \{v_j\in V\mid \{v_i,v_j\}\in E\}
$
be the neighborhood of node $v_i$ in $G$, and $n_{v_i} = |N_{v_i}|$ be its size. 
    
Consider an assignment of the agents to the nodes of $G$ chosen uniformly at random.
Let $W$ be a random variable denoting the social welfare of this assignment, $U_i$ a random variable denoting the expected utility of the agent placed at node $v_i$, and $X_i$ a binary random variable describing the color of this agent, where $X_i = 1$ if node $v_i$ is occupied by a blue agent and $X_i = 0$ if it is occupied by a red agent.
By the linearity of expectation and the law of total expectation, we have
\begin{align*}
\mathbb E[W] 
&= \sum_{i=1}^n \mathbb E[U_i] \\
&= \sum_{i=1}^n (\text{Pr}(X_i = 1)\cdot\mathbb E[U_i\mid X_i = 1]  + \text{Pr}(X_i = 0)\cdot\mathbb E[U_i\mid X_i = 0]).
\end{align*}
    
Now, for a fixed $v_i\in V$, it holds that
\begin{align*}
\mathbb E[U_i\mid X_i = 1]
&= \frac 1{n_{v_i}}\sum_{v_j\in N_{v_i}}\mathbb E[X_j\mid X_i = 1] \\
&= \frac 1{n_{v_i}}\sum_{v_j\in N_{v_i}} \text{Pr}(v_j \text{ blue}\mid v_i \text{ blue})= \frac 1{n_{v_i}}\sum_{v_j\in N_{v_i}} \frac{b-1}{n-1} = \frac{b-1}{n-1},
\end{align*}
where the first equality is again due to linearity of expectation.
Similarly, we have $\mathbb E[U_i\mid X_i = 0] = \frac{r-1}{n-1}$.
Hence, 
\begin{align*}\mathbb E[W] &= \sum_{i=1}^n \left(\frac bn\cdot\frac{b-1}{n-1} + \frac rn\cdot\frac{r-1}{n-1}\right) \\
&= b \cdot\frac{b-1}{n-1} + r\cdot \frac{r-1}{n-1} \\
&= \frac 1{n-1}(b(b-1) + (n-b)(n-b-1)) = \frac{1}{n-1}(n^2-n+2b(b-n)).
\end{align*}

Observe that the function $b(b-n)$ is decreasing in the range $b\in[0,n/2]$ and increasing in the range $b\in[n/2,n]$.
This means that for even $n$, we have
\begin{align*}
\mathbb E[W]&\ge \frac{1}{n-1}\left(n^2-n+2\cdot\frac{n}{2}\cdot\left(-\frac{n}{2}\right)\right) = \frac{n(n-2)}{2(n-1)} = g(n).
\end{align*}
For $n$ odd, since $b$ is an integer, it holds that
\begin{align*}
\mathbb E[W]&\ge \frac{1}{n-1}\left(n^2-n+2\cdot\frac{n-1}{2}\cdot\left(-\frac{n+1}{2}\right)\right) = \frac{n-1}{2} = g(n),
\end{align*}
implying that $\mathbb{E}[W] \ge g(n)$ in both cases.
Hence, there exists an assignment with social welfare at least $g(n)$.

Finally, it can be verified that when $G$ is a complete graph with $n$ nodes and the distribution of agents into types is balanced, every assignment has social welfare exactly $g(n)$.
\end{proof}

Next, we derandomize the algorithm in \Cref{thm:randomized} to produce an efficient deterministic algorithm that computes an assignment with welfare at least $g(n)$.
The pseudocode of the algorithm, which shares the notation of \Cref{thm:randomized}, can be found in \Cref{alg:approxWF}.
The main idea is that when we choose an agent to be assigned to an unassigned node, we pick a type such that the expected welfare is maximized, where the expectation is taken with respect to the uniform distribution of the remaining agents to the remaining nodes.

\begin{algorithm}[ht]

  \caption{Assignment with high social welfare}
  \label{alg:approxWF}
  \begin{flushleft}
    \textbf{Input:} Schelling instance $I = (N, G)$ with $G = (V,E)$\\
    \textbf{Output:} Assignment with social welfare at least $g(n)$
  \end{flushleft}

  \begin{algorithmic}[]

\FOR{$i = 1,\dots, n$}
    \IF{there is a unique assignment $\vv$ consistent with $X_1=a_1,\dots,X_{i-1}=a_{i-1}$ (up to permuting agents of the same color)}
    \RETURN $\vv$
    \ENDIF
    \STATE $W_0 = \mathbb E[W\mid X_1 = a_1,\dots, X_{i-1} = a_{i-1}, X_i =0]$
    \STATE $W_1 = \mathbb E[W\mid X_1 = a_1,\dots, X_{i-1} = a_{i-1}, X_i =1]$
        \IF{$W_1\ge W_0$}
        \STATE $a_i = 1$ /*assign a blue agent to $v_i$*/
        \ELSE 
        \STATE $a_i = 0$ /*assign a red agent to $v_i$*/
        \ENDIF
    \ENDFOR
  \RETURN Assignment corresponding to $(a_1,\dots,a_n)$
 \end{algorithmic}
\end{algorithm}

\begin{theorem}
\label{thm:derandomized}
\Cref{alg:approxWF} returns an assignment with social welfare at least $g(n)$ in polynomial time.
\end{theorem}

\begin{proof}
We use the same notation as in the proof of \Cref{thm:randomized}.

First, we prove that the welfare of the returned assignment is at least $g(n)$.
For $i = 0,\dots, n$, denote by $A_i$ the event $X_1 = a_1 \land X_2 = a_2 \land \dots \land X_i = a_i$. In particular, $A_0$ is the entire sample space.
We will show by induction that for each $i$, $\mathbb E[W\mid A_i]\ge \mathbb E[W]$.
The base case $i=0$ holds trivially.
For $i \in \{1,\dots, n\}$, if there is a unique assignment consistent with $X_1 = a_1 \land \dots \land X_{i-1} = a_{i-1}$, then the social welfare of the returned assignment is $\mathbb{E}[W\mid A_{i-1}]\geq \mathbb{E}[W]\geq g(n)$, where the first inequality follows from the induction hypothesis and the second inequality from \Cref{thm:randomized}.
Otherwise, we have
\begin{align*}
\mathbb E[W] &\le 
\mathbb E[W\mid A_{i-1}]\\ 
&=  \text{Pr}(X_i = 0\mid A_{i-1})\cdot \mathbb E[W\mid A_{i-1}\land X_i = 0] + \text{Pr}(X_i = 1\mid A_{i-1}) \cdot\mathbb E[W\mid A_{i-1}\land X_i = 1]\\
&\le  \text{Pr}(X_i = 0\mid A_{i-1}) \cdot\mathbb E[W\mid A_i] + \text{Pr}(X_i = 1\mid A_{i-1}) \cdot\mathbb E[W\mid A_i]
\\ &=  \mathbb E[W\mid A_i],
\end{align*}
where we use the law of total expectation for the first equality and the choice of $a_i$ in the algorithm for the second inequality.
This completes the induction.
Hence, if the algorithm terminates in the $j$th iteration, the welfare of the returned assignment is $\mathbb{E}[W\mid A_j]\geq \mathbb{E}[W]\geq g(n)$.

We next show that the algorithm can be implemented in polynomial time.
To this end, it suffices to show that the quantities $W_0$ and $W_1$ can be computed efficiently for each fixed $i\in\{1,\dots,n\}$.
If there is only one type of agents left after having assigned the first $i$ agents, this is straightforward, so assume that both types of agents still remain.
By the linearity of expectation, for each 
$x\in\{0,1\}$, 
\[
\mathbb{E}[W\mid A_{i-1}\land X_i = x] = \sum_{j=1}^n\mathbb{E}[U_j\mid A_{i-1}\land X_i = x].
\]
By the law of total expectation,
\begin{align*}
\mathbb E[U_j\mid A_{i-1}\land X_i=x] =
\text{Pr}&(X_j = 0\mid A_{i-1}\land X_i = x)\cdot \mathbb E[U_j\mid A_{i-1}\land X_i=x\land  X_j = 0] \\ 
&+ \text{Pr}(X_j = 1\mid A_{i-1}\land X_i = x) \cdot \mathbb E[U_j\mid A_{i-1}\land X_i=x\land X_j = 1],
\end{align*}
where a probability can be $0$ if $v_j$ has already been assigned an agent (i.e., if $j\leq i$). 
When $j > i$, we have
\[
\text{Pr}(X_j = 1\mid A_{i-1}\land X_i = x) = \frac{b - \sum_{k=1}^{i-1}a_k - x}{n-i}.
\]
Also, by the linearity of expectation,
\begin{align*}
\mathbb{E}[U_j&\mid A_{i-1}\land X_i=x\land X_j = 1] =  \frac 1{n_{v_j}} \sum_{v_k\in N_{v_j}}\mathbb E[X_k\mid A_{i-1}\land X_i=x\land X_j = 1].
\end{align*}
Finally, 
\begin{align*}
\mathbb E[X_k&\mid A_{i-1}\land X_i=x\land X_j = 1] = 
\begin{cases}
a_k \qquad &\text{ if } k\le i-1;\\
x \qquad &\text{ if } k = i;\\
\frac{b-\sum_{\ell=1}^{i-1}a_\ell - x-1}{n-i-1} \qquad &\text{ if } k > i.
\end{cases}
\end{align*}
The computations for $X_j = 0$ as well as for $j\le i$ can be done similarly.
\end{proof}

Since the social welfare of any assignment is at most $n$, \Cref{alg:approxWF} always produces an assignment with at least roughly half of the optimal welfare.
This raises the question of whether it is possible to compute a maximum-welfare assignment for any given instance in polynomial time.
Unfortunately, \citet{ElkindGaIg19} proved that maximizing the social welfare is NP-hard.
However, their proof relies on the existence of a ``stubborn agent'', who is assigned to a fixed node in advance and cannot move, and uses a topology with more nodes than agents.\footnote{\citet{AgarwalElGa20} showed that the hardness holds when the numbers of agents and nodes are equal, but still required stubborn agents and moreover assumed at least three types of agents.}
We show that the hardness remains even when both of these assumptions are removed and the topology is a regular graph, i.e., a graph in which all nodes have the same degree.

\begin{theorem}
\label{thm:SW-hardness}
The following problem is NP-complete: Given a Schelling instance and a rational number $s$,  decide whether there exists an assignment with social welfare at least $s$. The hardness holds even for the class of instances where the number of agents is equal to the number of nodes and the topology is a regular graph.
\end{theorem}

\begin{proof}
The problem belongs to NP since computing the social welfare of a given assignment can be done efficiently.
For the hardness, we reduce from the \textsc{Maximum Clique} problem for regular graphs, i.e., given a regular graph $G$ and an integer $k$, is there a clique of size at least $k$?
Note that this problem is NP-hard: Indeed, the \textsc{Independent Set} problem is NP-hard for regular graphs \citep[pp.~194--195]{GareyJo79}, and a set of vertices forms a clique in a given graph exactly when these vertices form an independent set in the complement graph.\footnote{\label{fn:complement}Given a graph $G = (V,E)$, its \emph{complement graph} is the graph $\overline G = (V, \overline E)$ with $\overline E = \{e \subseteq V \colon |e| = 2, e\notin E\}$, i.e., there is an edge between vertices $v_1,v_2\in V$ in $\overline G$ exactly when there is no edge between them in~$G$.}

Let $(G,k)$ be an instance of \textsc{Maximum Clique}, where $G = (V,E)$ is a $\rho$-regular graph on $n$ vertices, and $k$ an integer.
Define a Schelling instance on topology $G$ with $k$ red and $n-k$ blue agents.
For any assignment $\vv$, the social welfare $\SW(\vv)$ is equal to $n-2\delta(\vv)/\rho$, where $\delta(\vv)$ denotes the number of edges connecting a red agent and a blue agent in $\vv$.
If $G$ has a clique of size $k$, then by assigning all red agents to nodes in this clique, we have $\delta(\vv) = k\rho - 2\binom{k}{2}$; indeed, the sum of degrees of the red agents is $k\rho$, from which we have to subtract twice the number of red-red edges. 
Similarly, if $G$ does not have a clique of size $k$, then for any assignment $\vv$, we have $\delta(\vv) > k\rho - 2\binom{k}{2}$.
Hence, there exists an assignment with welfare at least $n-2k+4\binom{k}{2}/\rho$ if and only if $G$ has a clique of size $k$, so we may set $s = n-2k+4\binom{k}{2}/\rho$ to complete the reduction.
\end{proof}

Note that \Cref{thm:SW-hardness} also yields the hardness of computing a maximum-welfare assignment.
Indeed, any algorithm that computes such an assignment can also be used to decide whether there exists an assignment with a certain social welfare.

\section{Optimality Notions}\label{sec:notions}

Even when an assignment does not achieve maximum social welfare, there can still be other ways in which it is ``optimal''. In this section, we consider some optimality notions and quantify them in relation to social welfare.
We begin with a classic notion, Pareto optimality.

\begin{definition}
An assignment $\vv$ is said to be \emph{Pareto dominated} by an assignment $\vv'$ if $u_i(\vv)\le u_i(\vv')$ for all $i\in N$, with the inequality being strict for at least one agent.
An assignment $\vv$ is \emph{Pareto optimal} (PO) if it is not Pareto dominated by any other assignment.
\end{definition}

Given two vectors $\textbf{w}_1$ and $\textbf{w}_2$ of the same length $k$, we say that $\textbf{w}_1$ \emph{weakly dominates} $\textbf{w}_2$ if for each $i\in\{1,\dots,k\}$, the $i$th element of $\textbf{w}_1$ is at least that of $\textbf{w}_2$.
We say that $\textbf{w}_1$ \emph{strictly dominates} $\textbf{w}_2$ if at least one of the inequalities is strict.

For an assignment $\vv$, denote by $\textbf{u}(\vv)$ the vector of length $n$ consisting of the agents' utilities $u_i(\vv)$, sorted in non-increasing order.
Similarly, denote by $\textbf{u}_R(\vv)$ and $\textbf{u}_B(\vv)$ the corresponding sorted vectors of length $r$ and $b$ for the red and blue agents, respectively.
Note that an assignment $\vv$ is Pareto optimal if and only if there is no other assignment $\vv'$ such that $\textbf{u}_X(\vv')$ weakly dominates $\textbf{u}_X(\vv)$ for $X\in\{R,B\}$ and at least one of the dominations is strict.
Motivated by this observation, we define two new optimality notions appropriate for Schelling instances.

\begin{definition}
An assignment $\vv$ is said to be
\begin{itemize}
\item \emph{group-welfare dominated} by an assignment $\vv'$ if $\SW_X(\vv')\geq \SW_X(\vv)$ for $X\in\{R,B\}$ and at least one of the inequalities is strict;
\item \emph{utility-vector dominated} by an assignment $\vv'$ if $\textbf{u}(\vv')$ strictly dominates $\textbf{u}(\vv)$.
\end{itemize}
An assignment $\vv$ is \emph{group-welfare optimal} (GWO) if it is not group-welfare dominated by any other assignment. Similarly, an assignment $\vv$ is \emph{utility-vector optimal} (UVO) if it is not utility-vector dominated by any other assignment.
\end{definition}

\begin{figure}
\centering
\begin{tikzpicture}
\node (MW) {Maximum-welfare};
\node[above right =of MW] (GWO) {GWO};
\node[below right =of MW] (UVO) {UVO};
\node[below right =of GWO] (PO) {PO};
\draw[->] (MW) -- (GWO);
\draw[->] (MW) -- (UVO);
\draw[->] (GWO) -- (PO);
\draw[->] (UVO) -- (PO);
\draw[->] (MW) -- (PO);
\end{tikzpicture}
\caption{Implication relations among optimality notions.}
\label{fig:implications}
\end{figure}
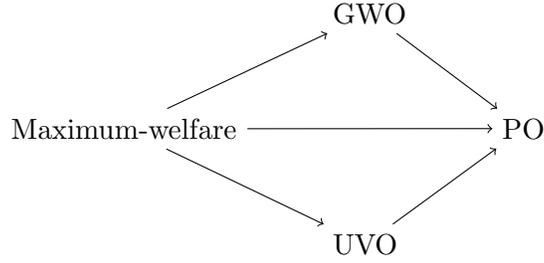

The implication relations in \Cref{fig:implications} follow immediately from the definitions; in particular, both of the new notions lie between welfare maximality and Pareto optimality.
We claim that no other implications exist between these notions.
To establish this claim, it suffices to show that GWO and UVO do not imply each other.

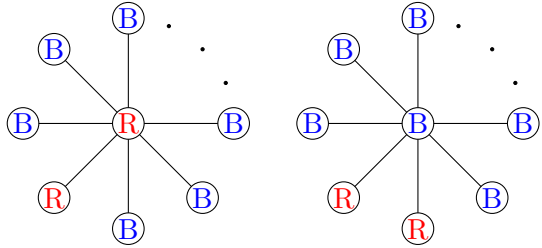
\begin{figure}
\centering
\begin{tikzpicture}[scale=0.7]

\draw  (2,2) -- (2,0);
\draw  (2,2) -- (2,4);
\draw  (2,2) -- (0,2);
\draw  (2,2) -- (4,2);
\draw  (2,2) -- (0.59,3.41);
\draw  (2,2) -- (3.41,0.59);
\draw  (2,2) -- (0.59,0.59);
\draw[fill=white] (2,2) circle [radius = 0.3];
\node[red] at (2,2) {R};
\draw[fill=white] (2,0) circle [radius = 0.3];
\node[blue] at (2,0) {B};
\draw[fill=white] (2,4) circle [radius = 0.3];
\node[blue] at (2,4) {B};
\draw[fill=white] (0,2) circle [radius = 0.3];
\node[blue] at (0,2) {B};
\draw[fill=white] (4,2) circle [radius = 0.3];
\node[blue] at (4,2) {B};
\draw[fill=white] (0.59,3.41) circle [radius = 0.3];
\node[blue] at (0.59,3.41) {B};
\draw[fill=white] (3.41,0.59) circle [radius = 0.3];
\node[blue] at (3.41,0.59) {B};
\draw[fill=white] (0.59,0.59) circle [radius = 0.3];
\node[red] at (0.59,0.59) {R};
\draw[fill=black] (3.41,3.41) circle [radius = 0.03];
\draw[fill=black] (3.85,2.77) circle [radius = 0.03];
\draw[fill=black] (2.77,3.85) circle [radius = 0.03];

\draw  (7.5,2) -- (7.5,0);
\draw  (7.5,2) -- (7.5,4);
\draw  (7.5,2) -- (5.5,2);
\draw  (7.5,2) -- (9.5,2);
\draw  (7.5,2) -- (6.09,3.41);
\draw  (7.5,2) -- (8.91,0.59);
\draw  (7.5,2) -- (6.09,0.59);
\draw[fill=white] (7.5,2) circle [radius = 0.3];
\node[blue] at (7.5,2) {B};
\draw[fill=white] (7.5,0) circle [radius = 0.3];
\node[red] at (7.5,0) {R};
\draw[fill=white] (7.5,4) circle [radius = 0.3];
\node[blue] at (7.5,4) {B};
\draw[fill=white] (5.5,2) circle [radius = 0.3];
\node[blue] at (5.5,2) {B};
\draw[fill=white] (9.5,2) circle [radius = 0.3];
\node[blue] at (9.5,2) {B};
\draw[fill=white] (6.09,3.41) circle [radius = 0.3];
\node[blue] at (6.09,3.41) {B};
\draw[fill=white] (8.91,0.59) circle [radius = 0.3];
\node[blue] at (8.91,0.59) {B};
\draw[fill=white] (6.09,0.59) circle [radius = 0.3];
\node[red] at (6.09,0.59) {R};
\draw[fill=black] (8.91,3.41) circle [radius = 0.03];
\draw[fill=black] (9.35,2.77) circle [radius = 0.03];
\draw[fill=black] (8.27,3.85) circle [radius = 0.03];
\end{tikzpicture}
\caption{Example showing that GWO does not imply UVO.}
\label{fig:star}
\end{figure}

\begin{proposition}
\label{prop:gwo-uvo}
GWO does not imply UVO.
\end{proposition}

\begin{proof}
Assume that the topology is a star as in \Cref{fig:star}, and there are two red and $n-2$ blue agents, where $n\ge 5$.
The left assignment $\vv$ is GWO, since putting a blue agent at the center as in the right assignment $\vv'$ leaves both red agents with utility $0$.
However, $\vv$ is not UVO, as 
\[
\textbf{u}(\vv) = (1, 1/(n-1), 0,\dots,0)
\] is strictly dominated by
\[
\textbf{u}(\vv') = (1, \dots, 1, (n-3)/(n-1), 0,0). \qedhere
\]
\end{proof}

\begin{figure}
\centering
\begin{tikzpicture}[scale=0.7]

\draw (4.5,10) -- (7.5,10) -- (4.5,8.5) -- (7.5,8.5) -- (4.5,7) -- (7.5,7) -- (4.5,5.5) -- (7.5,5.5) -- (4.5,3) -- (7.5,3) -- (4.5,1.5) -- (7.5,1.5) -- (4.5,3) -- (7.5,7) -- (4.5,8.5) -- (7.5,5.5) -- (4.5,1.5) -- (7.5,7) -- (4.5,10) -- (7.5,8.5) -- (4.5,5.5) -- (7.5,3) -- (4.5,7) -- (7.5,10) -- (4.5,5.5) -- (7.5,1.5) -- (4.5,7) -- (7.5,5.5) -- (4.5,10) -- (7.5,3) -- (4.5,8.5) -- (7.5,1.5) -- (4.5,10);
\draw (7.5,10) -- (4.5,1.5) -- (7.5,8.5) -- (4.5,3) -- (7.5,10);
\draw[fill=white] (4.5,10) circle [radius = 0.3];
\node[red] at (4.5,10) {R};
\draw[fill=white] (4.5,8.5) circle [radius = 0.3];
\node[blue] at (4.5,8.5) {B};
\draw[fill=white] (4.5,7) circle [radius = 0.3];
\node[blue] at (4.5,7) {B};
\draw[fill=white] (4.5,5.5) circle [radius = 0.3];
\node[blue] at (4.5,5.5) {B};
\draw[fill=white] (4.5,3) circle [radius = 0.3];
\node[blue] at (4.5,3) {B};
\draw[fill=white] (4.5,1.5) circle [radius = 0.3];
\node[blue] at (4.5,1.5) {B};
\draw[fill=black] (4.5,4.75) circle [radius = 0.03];
\draw[fill=black] (4.5,4.25) circle [radius = 0.03];
\draw[fill=black] (4.5,3.75) circle [radius = 0.03];
\draw[fill=white] (7.5,10) circle [radius = 0.3];
\node[blue] at (7.5,10) {B};
\draw[fill=white] (7.5,8.5) circle [radius = 0.3];
\node[red] at (7.5,8.5) {R};
\draw[fill=white] (7.5,7) circle [radius = 0.3];
\node[red] at (7.5,7) {R};
\draw[fill=white] (7.5,5.5) circle [radius = 0.3];
\node[red] at (7.5,5.5) {R};
\draw[fill=white] (7.5,3) circle [radius = 0.3];
\node[red] at (7.5,3) {R};
\draw[fill=white] (7.5,1.5) circle [radius = 0.3];
\node[red] at (7.5,1.5) {R};
\draw[fill=black] (7.5,4.75) circle [radius = 0.03];
\draw[fill=black] (7.5,4.25) circle [radius = 0.03];
\draw[fill=black] (7.5,3.75) circle [radius = 0.03];

\draw (9.5,10) -- (12.5,10) -- (9.5,8.5) -- (12.5,8.5) -- (9.5,7) -- (12.5,7) -- (9.5,5.5) -- (12.5,5.5) -- (9.5,3) -- (12.5,3) -- (9.5,1.5) -- (12.5,1.5) -- (9.5,3) -- (12.5,7) -- (9.5,8.5) -- (12.5,5.5) -- (9.5,1.5) -- (12.5,7) -- (9.5,10) -- (12.5,8.5) -- (9.5,5.5) -- (12.5,3) -- (9.5,7) -- (12.5,10) -- (9.5,5.5) -- (12.5,1.5) -- (9.5,7) -- (12.5,5.5) -- (9.5,10) -- (12.5,3) -- (9.5,8.5) -- (12.5,1.5) -- (9.5,10);
\draw (12.5,10) -- (9.5,1.5) -- (12.5,8.5) -- (9.5,3) -- (12.5,10);
\draw[fill=white] (9.5,10) circle [radius = 0.3];
\node[red] at (9.5,10) {R};
\draw[fill=white] (9.5,8.5) circle [radius = 0.3];
\node[blue] at (9.5,8.5) {B};
\draw[fill=white] (9.5,7) circle [radius = 0.3];
\node[red] at (9.5,7) {R};
\draw[fill=white] (9.5,5.5) circle [radius = 0.3];
\node[blue] at (9.5,5.5) {B};
\draw[fill=white] (9.5,3) circle [radius = 0.3];
\node[red] at (9.5,3) {R};
\draw[fill=white] (9.5,1.5) circle [radius = 0.3];
\node[blue] at (9.5,1.5) {B};
\draw[fill=black] (9.5,4.75) circle [radius = 0.03];
\draw[fill=black] (9.5,4.25) circle [radius = 0.03];
\draw[fill=black] (9.5,3.75) circle [radius = 0.03];
\draw[fill=white] (12.5,10) circle [radius = 0.3];
\node[red] at (12.5,10) {R};
\draw[fill=white] (12.5,8.5) circle [radius = 0.3];
\node[blue] at (12.5,8.5) {B};
\draw[fill=white] (12.5,7) circle [radius = 0.3];
\node[red] at (12.5,7) {R};
\draw[fill=white] (12.5,5.5) circle [radius = 0.3];
\node[blue] at (12.5,5.5) {B};
\draw[fill=white] (12.5,3) circle [radius = 0.3];
\node[red] at (12.5,3) {R};
\draw[fill=white] (12.5,1.5) circle [radius = 0.3];
\node[blue] at (12.5,1.5) {B};
\draw[fill=black] (12.5,4.75) circle [radius = 0.03];
\draw[fill=black] (12.5,4.25) circle [radius = 0.03];
\draw[fill=black] (12.5,3.75) circle [radius = 0.03];
\end{tikzpicture}
\caption{Example showing that UVO does not imply GWO. The topology is a complete bipartite graph.}
\label{fig:complete-bipartite}
\end{figure}
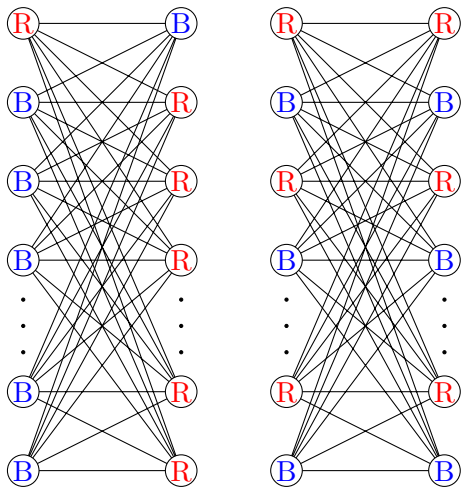

\begin{proposition}
\label{prop:uvo-gwo}
UVO does not imply GWO.
\end{proposition}

\begin{proof}
Let $n$ be a multiple of $4$. Suppose that the topology is a complete bipartite graph with $n/2$ nodes on each side, and there are $n/2$ red and $n/2$ blue agents (\Cref{fig:complete-bipartite}).
The left assignment $\vv$, which assigns one red agent to the left side and one blue agent to the right side, is UVO.
Indeed, the red agent assigned to the left side receives utility $(n/2-1)/(n/2)$, and any assignment in which an agent receives equal or higher utility must have the same sorted utility vector as $\vv$.
We have
\[
\SW(\vv) = 2\cdot\frac{n/2-1}{n/2} + 2(n/2-1)\cdot\frac{1}{n/2} = 4 - \frac{8}{n},
\]
with each group receiving half of the welfare, i.e., $2 - 4/n$.
On the other hand, in the right assignment $\vv'$, which assigns half of the agents of each color to each side, every agent receives utility $1/2$.
Hence $\SW(\vv') = n/2$, and each group receives a total utility of $n/4$. 
It follows that when $n\ge 8$, $\vv$ is UVO but not GWO.
\end{proof}

In order to quantify the welfare guarantee that each optimality notion provides, we define the price of a notion as follows.

\begin{definition}
Given a property $P$ of assignments and a Schelling instance, the \emph{price of $P$} for that instance is defined as the ratio between the maximum social welfare (of any assignment) and the minimum social welfare of an assignment satisfying $P$:
\[
\text{Price of $P$ for instance $I$} = \frac{\SW(\vv^*(I))}{\min_{\vv\in P(I)}\SW(\vv)},
\]
where $P(I)$ is the set of all assignments satisfying $P$ in instance $I$.\footnote{We interpret the ratio $\frac{0}{0}$ in this context to be equal to $1$.}
The \emph{price of $P$} for a class of instances is then defined as the supremum price of $P$ over all instances in that class.
\end{definition}

For $P\in\{$PO, GWO, UVO$\}$, we have $\vv^*(I)\in P(I)$, so the price of $P$ is always well-defined and at least $1$.
Note also that $\max_{\vv\in P(I)}\SW(\vv) = \SW(\vv^*(I))$.

In \Cref{fig:star}, the left assignment is GWO and PO and has social welfare $n/(n-1)$, whereas the maximum-welfare assignment on the right has social welfare $n(n-3)/(n-1)$.
We therefore have the following bound (for $n\le 4$, the bound holds trivially).

\begin{proposition}
\label{prop:star-bound}
For every $n$, both the price of GWO and the price of PO are at least $n-3$.
\end{proposition}

The following result shows that the welfare of a UVO assignment can also be a linear factor away from the maximum welfare, but not more.

\begin{theorem}
\label{thm:uvo-price}
The price of UVO is $\Theta(n)$.
\end{theorem}

\begin{proof}
We prove the lower and the upper bound separately.

\medskip

\noindent
\underline{Lower bound}:
Consider the topology in \Cref{fig:complete-bipartite}.
As in the proof of \Cref{prop:uvo-gwo}, the left assignment $\vv$ is UVO and has social welfare $4-8/n$.
On the other hand, the right assignment $\vv'$ has social welfare $n/2$, meaning that the ratio $\SW(\vv')/\SW(\vv)$ is greater than $n/8$.

\medskip

\noindent
\underline{Upper bound}:
We claim that if $n\ge 3$, any UVO assignment has social welfare at least\footnote{In \Cref{thm:UVO-welfare-1}, we improve this bound to $1$ via a longer proof.} $1/2$; since the maximum social welfare is at most $n$, this yields the desired bound.

Assume first that the number of agents is equal to the number of nodes.
Let $\vv$ be a UVO assignment.
If there is a red agent and a blue agent both receiving utility $0$, then since no node is empty and $n\ge 3$, swapping them yields an improvement with respect to the utility vector.
So we may assume that all agents of one type, say blue, receive a positive utility.
If at least $n/2$ agents receive a positive utility, then $\SW(\vv)\ge n/(2n-2) > n/(2n) = 1/2$.
Assume therefore that more than $n/2$ agents receive utility $0$; these agents must all be red.
Swap $b$ of these red agents receiving utility $0$ with all $b$ blue agents to obtain an assignment~$\vv'$.
Notice that the utility in $\vv'$ of each of these $b$ red agents is at least as high as the utility of the blue agent in $\vv$ with whom she was swapped, while all blue agents receive utility $0$ in $\vv'$.
In addition, every other (red) agent is not worse off, and at least one of them is better off (in particular, one who receives utility $0$ in $\vv$, which must exist since $n/2 > b$).
Hence $\vv$ is utility-vector dominated by $\vv'$, a contradiction.

Now, assume that the number of agents is less than the number of nodes.
Since $n\ge 3$, any UVO assignment $\vv$ must have $\SW(\vv) > 0$, so there exists a connected component (of the topology restricted to the nodes occupied according to $\vv$) with a positive social welfare.
Let $n'$ be the size of this component.
If $n'=2$, then $\SW(\vv)\ge 2$.
Else, the assignment restricted to this component is also UVO, and by our earlier arguments has social welfare at least $1/2$.
\end{proof}

Next, we show that the price of GWO is also $\Theta(n)$. The lower bound follows from Proposition~\ref{prop:star-bound}, while the upper bound follows from the fact that the social welfare never exceeds $n$ along with the following theorem, which establishes a lower bound of (at least) $1$ on the social welfare of GWO assignments.

\begin{theorem}
\label{thm:gwo-welfare}
Any GWO assignment has social welfare at least $n/(n-1)$ for $n\ge 4$, and $1$ for $n=3$.
Moreover, these bounds cannot be improved.
\end{theorem}

\begin{proof}
To see that the bounds cannot be improved, consider the left assignment in \Cref{fig:star} for $n\ge 4$, and a triangle topology with two red and one blue agents for $n=3$.

Assume first that the number of agents is equal to the number of nodes.
The case $n=3$ can be verified directly, since the only two possible topologies are a triangle and a path.
Let $n\ge 4$, and assume for contradiction that there exists a GWO assignment $\vv$ with social welfare less than $n/(n-1)$.
Since the least possible positive utility of an agent is $1/(n-1)$, this means that some agent receives utility $0$.
Take such an agent $i$, and assume without loss of generality that $i$ is red.
Since the numbers of agents and nodes are equal, $i$ is connected to a set $A\neq\emptyset$ of blue agents.
Consider the following cases.

\medskip

\noindent
\underline{Case 1}: There is a blue agent $j$ outside $A$.
We swap $i$ and $j$ to obtain an assignment $\vv'$.
After the swap, $j$ has utility $1$, and each blue agent in $A$ has utility at least $1/(n-1)$, so $\SW_B(\vv')\ge n/(n-1) > \SW(\vv)\ge \SW_B(\vv)$.
In addition, no red agent receives a lower utility in $\vv'$ than in $\vv$.
Hence $\vv$ is not GWO, a contradiction.

\medskip

\noindent
\underline{Case 2}: There are no blue agents outside $A$, but at least one red agent besides $i$.
Since no node is empty, there is a blue agent $j\in A$ who is adjacent to another red agent.
We swap $i$ and $j$ to obtain an assignment $\vv'$.
No agent receives a lower utility in $\vv'$ than in $\vv$.
Moreover, $i$'s utility strictly increases. 
Hence $\vv$ is not GWO, a contradiction.

\medskip

\noindent
\underline{Case 3}: There are no blue agents outside $A$, and no red agent besides $i$.
This means that $i$ receives utility $0$ in any assignment.
Consider an assignment $\vv'$ such that the blue agents form a connected component.
Each blue agent receives utility at least $1/2$, so $\SW_B(\vv')\ge b/2 \ge 3/2 > n/(n-1)$.
It follows that $\vv$ is group-welfare dominated by $\vv'$, a contradiction.

\medskip

Now, assume that the number of agents is less than the number of nodes. Since $n\ge 3$, any GWO assignment $\vv$ must have $\SW(\vv)>0$, so there exists a connected component (of the topology restricted to the nodes occupied according to $\vv$) with a positive social welfare. Since the assignment restricted to this component is GWO, we are done if the component is of size at least $n'\ge 4$, because then $\SW(\vv)\ge n'/(n'-1)\ge n/(n-1)$. 
If there is a component of size $2$ with positive welfare, the social welfare is at least $2$ and we are also done.
Since any component of positive welfare has a welfare of at least $1$, we are again done if there are at least two components of positive welfare. 
Thus, we can assume that there is a single component of positive welfare of size $3$, which we can moreover assume has the topology of a triangle (the only other possibility being a path, which guarantees a welfare of $3/2$), and that there are exactly two agents of one type, say blue, and one agent of the other type in this component.

If there is another blue agent outside this component, $\vv$ is group-welfare dominated by the assignment that swaps the red agent of the triangle and this blue agent. Finally, consider the case where there is another red agent and no blue agent outside the triangle.
The red agent must have an empty neighboring node. We obtain a group-welfare improvement by moving the red agent of the triangle to this empty node. Hence, the remaining case is that all agents are part of the triangle, meaning that $n=3$; in this case, the social welfare is $1$, as desired.
\end{proof}

We now turn to Pareto optimality, for which we prove a weaker lower bound on the social welfare.

\begin{theorem}
\label{thm:PO-welfare}
When $n\ge 3$, any PO assignment has social welfare at least $1/\sqrt{n}$.
\end{theorem}

\begin{proof}
By an argument similar to that in the upper bound part of \Cref{thm:uvo-price}, and since the function $1/\sqrt{n}$ is decreasing, it suffices to consider the case where the number of agents is equal to the number of nodes.
Let $\vv$ be a PO assignment.
If at least $\sqrt{n}$ agents receive a positive utility, then $\SW(\vv)\ge \sqrt{n}/(n-1) \ge 1/\sqrt{n}$, so assume that fewer than $\sqrt{n}$ agents receive a positive utility.
Similarly, we may assume that every agent receives utility less than $1/\sqrt{n}$.
If there is a red agent and a blue agent both receiving utility $0$, then since no node is empty and $n\ge 3$, swapping them yields a Pareto improvement.
So we may assume that all agents of one type, say blue, receive a positive utility.
This means in particular that $b < \sqrt{n}$, and more than $n-\sqrt{n}$ red agents only have blue neighbors.
Hence, there exists a blue agent $i$ with at least $(n-\sqrt{n})/\sqrt{n} = \sqrt{n}-1\ge b-1$ red neighbors.
Let $A$ be a set containing $b-1$ of these red neighbors.

Swap the $b-1$ blue agents other than $i$ with the red agents in $A$ to obtain an assignment $\vv'$.
Since these red agents are not adjacent to any red agent in $\vv$, no red agent is worse off in $\vv'$.
Each of the $b-1$ blue agents receives utility at least $1/b > 1/\sqrt{n}$ in $\vv'$, so all of them are strictly better off.
Furthermore, $i$ is adjacent to all of these $b-1$ blue agents in $\vv'$, and therefore cannot be worse off.
Hence $\vv$ is not PO, a contradiction.
\end{proof}

Combined with Proposition~\ref{prop:star-bound}, \Cref{thm:PO-welfare} implies that when $n\ge 3$, the price of PO is at least $n-3$ and at most $n\sqrt{n}$. We conjecture that the welfare guarantee in \Cref{thm:PO-welfare} can be improved to $n/(n-1)$ for $n\ge 4$, which would be tight due to the left assignment in \Cref{fig:star}.
Next, we confirm this conjecture when the topology is a tree.

\begin{figure}
\centering
\begin{tikzpicture}[scale=0.7]
\draw (6,9) -- (6,6) -- (4,4.5) -- (3,3);
\draw (6,9) -- (3,7.5) -- (3,6);
\draw (5,3) -- (4,4.5) -- (4,3);
\draw (6,6) -- (8,4.5) -- (7,3);
\draw (8,4.5) -- (9,3);
\draw[fill=white] (6,9) circle [radius = 0.3];
\node[red] at (6,9) {R};
\draw[fill=white] (6,7.5) circle [radius = 0.3];
\node[red] at (6,7.5) {R};
\draw[fill=white] (3,7.5) circle [radius = 0.3];
\node[blue] at (3,7.5) {B};
\draw[fill=white] (6,6) circle [radius = 0.3];
\node[blue] at (6,6) {B};
\draw[fill=white] (3,6) circle [radius = 0.3];
\node[red] at (3,6) {R};
\draw[fill=white] (4,4.5) circle [radius = 0.3];
\node[blue] at (4,4.5) {B};
\draw[fill=white] (8,4.5) circle [radius = 0.3];
\node[blue] at (8,4.5) {B};
\draw[fill=white] (3,3) circle [radius = 0.3];
\node[red] at (3,3) {R};
\draw[fill=white] (4,3) circle [radius = 0.3];
\node[red] at (4,3) {R};
\draw[fill=white] (5,3) circle [radius = 0.3];
\node[red] at (5,3) {R};
\draw[fill=white] (7,3) circle [radius = 0.3];
\node[red] at (7,3) {R};
\draw[fill=white] (9,3) circle [radius = 0.3];
\node[red] at (9,3) {R};
\node at (10.5,3) {Level 1};
\node at (10.5,4.5) {Level 2};
\node at (10.5,6) {Level 3};
\node at (10.5,7.5) {Level 4};
\node at (10.5,9) {Level 5};
\node at (3.45,4.5) {$i$};
\node at (5.45,6) {$j$};
\node at (5.45,7.5) {$k$};
\end{tikzpicture}
\caption{Illustration for the proof of \Cref{thm:PO-welfare-tree}.}
\label{fig:PO-welfare-tree}
\end{figure}
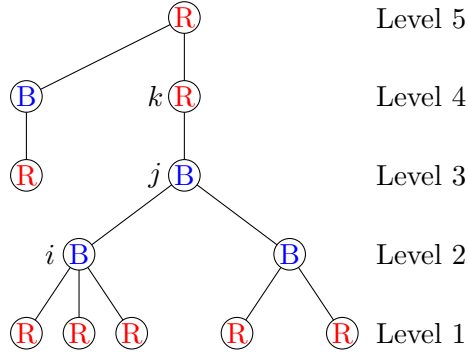

\begin{theorem}
\label{thm:PO-welfare-tree}
When $n\ge 3$ and the topology is a tree, any PO assignment has social welfare at least $n/(n-1)$.
Moreover, this bound cannot be improved.
\end{theorem}

\begin{proof}
The bound cannot be improved due to the left assignment in \Cref{fig:star}.
To establish the bound, note first that by an argument similar to that in the upper bound part of \Cref{thm:uvo-price}, and since the function $n/(n-1)$ is decreasing, it suffices to consider the case where the number of agents is equal to the number of nodes.
Let $\vv$ be a PO assignment.
As in the proof of \Cref{thm:PO-welfare}, we may assume that all agents of one type, say blue, receive a positive utility.
If an agent occupying a leaf node is adjacent to another agent of the same type, then $\SW(\vv)\ge 1+1/(n-1)$ and we are done.
Hence, assume that all leaf nodes are occupied by red agents receiving utility $0$.

Root the tree at an arbitrary node.
The deepest level of the tree consists only of leaf nodes; call this ``level 1'', and call the other levels accordingly (see \Cref{fig:PO-welfare-tree}).
Consider an arbitrary parent $i$ of a leaf at level 1---this parent must be blue.
Since all blue agents receive a positive utility, $i$ must have a blue parent $j$.
Any branch originating from $j$ has to go down to level 1; otherwise, the branch stops at level 2 with a red leaf, and we can swap this red leaf with $i$ to obtain a Pareto improvement.
In particular, in $j$'s subtree, all nodes on level 2 are blue.
If $j$ does not have a parent or has a blue parent, then she has utility $1$ and $\SW(\vv)\ge n/(n-1)$.
Assume therefore that $j$ has a red parent $k$. 
Note that $j$ receives utility at least $1/2$, and any of its child at least $1/(n-1)$.

Suppose that $k$ has another child. 
If $k$ has a red child, this child cannot be a leaf (because all leaves have utility $0$), so it must in turn have a blue child (otherwise it receives utility $1$ and we are done), but then this blue child receives utility $0$, a contradiction.
Hence $k$ can only have blue children.
Suppose that $k$ has a blue child $\ell\ne j$, which cannot be a leaf.
If $\ell$ only has blue children, it receives utility at least $1/2$, and we are done since $\SW(\vv)\ge 1/2 + 1/(n-1) + 1/2 = n/(n-1)$.
So assume that $\ell$ has a red child $m$, which must also be a leaf.
Since all blue agents receive a positive utility, $\ell$ must also have a blue child $o$, which must in turn have only red children.
Now, swapping $m$ and $o$ leads to a Pareto improvement, a contradiction.
Hence we may assume that $k$'s only child is $j$.

Finally, if $k$ does not have a parent or has a blue parent, swapping $k$ with $i$ yields a Pareto improvement.
So assume that $k$ has a red parent.
This means that $k$ receives utility $1/2$.
Combining this with the utility of $j$ and her children, we again have $\SW(\vv)\ge n/(n-1)$.
\end{proof}

Together with Proposition~\ref{prop:star-bound}, which holds for trees, \Cref{thm:PO-welfare-tree} gives a tight bound on the price of PO for trees.

\begin{corollary}
\label{cor:PO-price-tree}
When the topology is a tree, the price of PO is $\Theta(n)$.
\end{corollary}

To finish this section, we show that if $b/r\in \Theta(1)$, i.e., the fraction of agents of each type is at least a certain constant, then a constant welfare can again be guaranteed.

\begin{proposition}
\label{prop:b/r-constant}
Suppose that the number of agents is equal to the number of nodes.
Then any PO assignment has social welfare at least
$\min\left\{\frac{b}{r+1}, \frac{r}{b+1}\right\}$.
\end{proposition}

\begin{proof}
Let $\vv$ be a PO assignment.
Since the number of agents is equal to the number of nodes, as in the proof of \Cref{thm:PO-welfare}, we may assume that all agents of one type, say blue, receive a positive utility.
Since a blue agent is adjacent to at least one other blue agent and at most $r$ red agents, each blue agent receives utility at least $1/(r+1)$.
This means that $\SW(\vv)\ge b/(r+1)$.
Similarly, if all red agents receive a positive utility, then $\SW(\vv)\ge r/(b+1)$.
The conclusion follows.
\end{proof}

When the ratio $b/r$ is upper and lower bounded by constants, the welfare guarantee provided by \Cref{prop:b/r-constant} is also constant, which is at most a linear factor away from the maximum welfare.
Since \Cref{fig:complete-bipartite} shows an instance with $b=r$ and a PO assignment whose welfare is a linear factor away from the maximum welfare, we obtain the following:

\begin{corollary}
\label{cor:b/r-constant}
When $b/r\in\Theta(1)$ and the number of agents is equal to the number of nodes, the price of PO is $\Theta(n)$.
\end{corollary}

\section{Computing Optimal Assignments}\label{sec:compopt}

In the previous section, we introduced two new concepts of optimality and studied the welfare guarantees that they provide along with Pareto optimality.
In this section, we continue our investigation of these optimality notions by examining the complexity of computing assignments satisfying them. 
Furthermore, we consider other common welfare notions such as egalitarian welfare and Nash welfare.

First, we observe that in the reduced instances of \Cref{thm:SW-hardness}, an assignment is GWO if and only if it maximizes the social welfare.\footnote{To see this, note that for any assignment $\vv$, the sum of utilities of the red and blue agents is equal to $\SW_R(\vv) = r -\delta(\vv)/\rho$ and $\SW_B(\vv) = b -\delta(\vv)/\rho$, respectively, where $\delta(\vv)$ denotes the number of edges connecting a red agent and a blue agent in $\vv$. Hence, both the GWO and maximum-welfare assignments are precisely the assignments minimizing the number of these interconnection edges.}
This immediately yields the following intractability.

\begin{theorem}
\label{thm:GWO-hardness}
Computing a GWO assignment is NP-hard, even for the class of Schelling instances where the number of agents is equal to the number of nodes and the topology is a regular graph.
\end{theorem}

For the other two optimality notions, we will establish a close relationship to the problem of computing a ``perfect assignment'', wherein every agent receives utility $1$.

\begin{definition}
An assignment $\vv$ is called \emph{perfect} if $u_i(\vv) = 1$ for all $i\in N$.
\end{definition}

We start by showing the hardness of this problem.

\begin{figure}
\centering
\begin{tikzpicture}[scale=0.7]
\draw (2,2) -- (5,3.5) -- (2,3.5);
\draw (5,3.5) -- (2,5);
\draw (2,6.5) -- (5,8) -- (2,8);
\draw (5,8) -- (2,9.5);
\draw (5,8) -- (6,9.5) -- (8,9.5) -- (8,6.5) -- (6,6.5) -- (5,8) -- (8,9.5) -- (6,6.5) -- (6,9.5) -- (8,6.5) -- (5,8);
\draw (5,8) -- (9,8) -- (12,5.75);
\draw (5,3.5) -- (6,5) -- (8,5) -- (8,2) -- (6,2) -- (5,3.5) -- (8,5) -- (6,2) -- (6,5) -- (8,2) -- (5,3.5);
\draw (5,3.5) -- (9,3.5) -- (12,5.75);
\draw[fill=white] (2,2) circle [radius = 0.3];
\draw[fill=white] (2,3.5) circle [radius = 0.3];
\draw[fill=white] (2,5) circle [radius = 0.3];
\draw[fill=white] (2,6.5) circle [radius = 0.3];
\draw[fill=white] (2,8) circle [radius = 0.3];
\draw[fill=white] (2,9.5) circle [radius = 0.3];
\draw[fill=white] (5,8) circle [radius = 0.3];
\draw[fill=white] (6,9.5) circle [radius = 0.3];
\draw[fill=white] (8,9.5) circle [radius = 0.3];
\draw[fill=white] (9,8) circle [radius = 0.3];
\draw[fill=white] (8,6.5) circle [radius = 0.3];
\draw[fill=white] (6,6.5) circle [radius = 0.3];
\draw[fill=white] (5,3.5) circle [radius = 0.3];
\draw[fill=white] (6,5) circle [radius = 0.3];
\draw[fill=white] (8,5) circle [radius = 0.3];
\draw[fill=white] (9,3.5) circle [radius = 0.3];
\draw[fill=white] (8,2) circle [radius = 0.3];
\draw[fill=white] (6,2) circle [radius = 0.3];
\draw[fill=white] (12,5.75) circle [radius = 0.3];
\node at (1.4,9.5) {$1$};
\node at (1.4,8) {$2$};
\node at (1.4,6.5) {$3$};
\node at (1.4,5) {$4$};
\node at (1.4,3.5) {$5$};
\node at (1.4,2) {$6$};
\node at (5,2.8) {$y_1$};
\node at (5.35,5) {$y_2$};
\node at (8.65,5) {$y_3$};
\node at (6,1.4) {$y_4$};
\node at (8,1.4) {$y_5$};
\node at (9.5,3) {$y_6$};
\node at (5,8.65) {$x_1$};
\node at (6,10.1) {$x_2$};
\node at (8,10.1) {$x_3$};
\node at (5.35,6.5) {$x_4$};
\node at (8.65,6.5) {$x_5$};
\node at (9.5,8.5) {$x_6$};
\node at (12.55,5.75) {$z$};
\end{tikzpicture}
\caption{Illustration for the proof of Theorem~\ref{thm:perfect} when $R=\{1,2,3,4,5,6\}$ and $S=\{x,y\}$ where $x=\{1,2,3\}$ and $y=\{4,5,6\}$.}
\label{fig:perfect}
\end{figure}
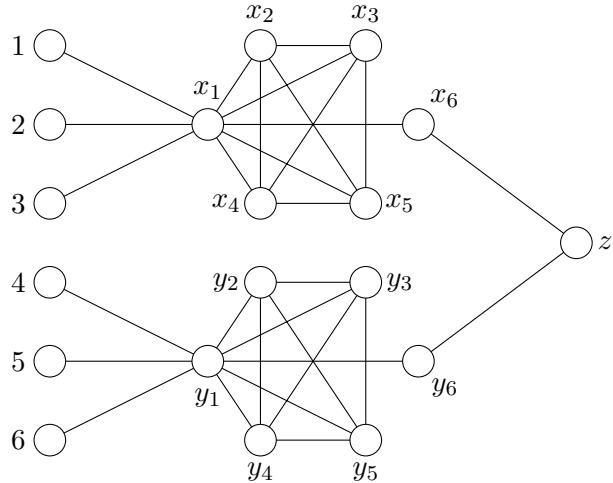

\begin{theorem}\label{thm:perfect}
Deciding whether there exists a perfect assignment is NP-complete.
\end{theorem}

\begin{proof}
Membership in NP is clear: a perfect assignment can be verified in polynomial time.

The hardness reduction is from \textsc{Exact 3-Cover (X3C)}. An instance of \textsc{X3C} consists of a tuple $(R,S)$, where $R$ is a ground set whose size is divisible by $3$, and $S$ is a collection of $3$-element subsets of $R$, where $|S|\ge |R|/3$. 
A Yes-instance is an instance in which there exists a subcollection $S'\subseteq S$ of size $|R|/3$ that exactly partitions~$R$.
It is well-known that \textsc{X3C} is NP-hard \citep[p.~221]{GareyJo79}.

Let an instance $(R,S)$ of \textsc{X3C} be given. We define a Schelling instance with $1 + |S| - |R|/3$ blue and $|R| + 2|S|^2 + |R|/3$ red agents. 
Define the topology graph $G = (V,E)$ with $V = R\cup \{s_i\colon 1\le i\le 2|S| +2, s\in S\}\cup \{z\}$ and edges given by

\begin{itemize}
    \item $\{r,s_1\}\in E$ if $r\in R$ and $s\in S$ with $r\in s$;
    \item $\{s_i,s_j\}\in E$ for $s\in S$, $1\le i,j\le 2|S|+1$;
    \item $\{s_1,s_{2|S|+2}\},\{s_{2|S|+2},z\}\in E$ for $s\in S$; and
    \item no further edges are in $E$.
\end{itemize}
See Figure~\ref{fig:perfect} for an illustration.
Note that the number of nodes is $|R|+2|S|^2+2|S|+1$ and the number of agents is $|R|+2|S|^2+|S|+1$, so exactly $|S|$ nodes are left empty in any assignment.
We claim that $(R,S)$ is a Yes-instance if and only if there exists a perfect assignment in the Schelling instance.

Assume first that $(R,S)$ is a Yes-instance, and let $S'\subseteq S$ be a partition of $R$ using sets in $S$. 
We assign the blue agents to the nodes in the set $A = \{z\}\cup \{s_{2|S|+2}\colon s\in S\setminus S'\}$, and the red agents to the vertices in the set $B = R\cup \{s_i\colon 2\le i\le 2|S| +1, s\in S\}\cup \{s_1\colon s\in S'\}$. 
Note that $|A| = 1 + |S| - |R|/3$, $B = |R|  + 2|S|^2 + |R|/3$, and the nodes in $A$ induce a connected subgraph of $G$ that has no neighbor in $B$. 
This means that all blue agents receive utility $1$.
Since $S'$ covers $R$, all red agents also receive utility $1$, meaning that the assignment is perfect.

Conversely, assume that there is a perfect assignment. 
Since every assignment leaves exactly $|S|$ nodes empty, no blue agent can be assigned to a vertex in $\{s_i\colon 1\le i\le 2|S| +1, s\in S\}$, because she would then have a red neighbor. 
Additionally, no blue agent can be assigned to a vertex in $R$, because then some of her only neighbors in the set $\{s_1\colon s\in S\}$ would also have to be blue, which is impossible by the previous sentence. 
We conclude that the blue agents are assigned to the nodes in $\{z\}\cup \{s_{2|S|+2}\colon s\in S\}$, which means in particular that some blue agent is assigned to $z$. 
Define $S' = \{s\in S\colon s_1 \textnormal{ is red}\}$. Then, the empty nodes are precisely $\{s_{2|S|+2}\colon s\in S'\}\cup\{s_1 \colon s \in S\setminus S'\}$, and we have $|S'| = |R|/3$. 
In particular, all nodes in $R$ must have red agents. 
Now, such an agent can receive a positive utility only if at least one of her neighbors is red. 
Hence, $S'$ covers $R$. 
Since $|S'| = |R|/3$, it follows that $S'$ forms a partition of $R$. 
\end{proof}

\Cref{thm:perfect} turns out to be particularly useful for deriving hardness results with respect to other optimality and welfare notions.

\begin{corollary}
Computing a UVO assignment (resp., PO assignment) is NP-hard.
\end{corollary}

\begin{proof}
Observe that if there exists a perfect assignment in an instance, then every UVO (resp., PO) assignment in that instance must be perfect.
Hence, an algorithm for computing a UVO (resp., PO) assignment can be used to decide whether a perfect assignment exists.
The conclusion follows from \Cref{thm:perfect}.
\end{proof}

In addition, we obtain the hardness of computing assignments with maximum egalitarian or Nash welfare. 
Recall that the \emph{egalitarian welfare} of an assignment is the minimum among the agents' utilities in that assignment, and the \emph{Nash welfare} is the product of the agents' utilities in that assignment.

\begin{corollary}
The following problem is NP-complete: Given a Schelling instance and a rational number $s$, decide whether there exists an assignment with egalitarian (resp., Nash) welfare at least $s$.
\end{corollary}

\begin{proof}
Membership in NP is trivial.
For the hardness, observe that an assignment is perfect if and only if its egalitarian (resp., Nash) welfare is (at least) $1$, so the conclusion follows from \Cref{thm:perfect}.
\end{proof}

We emphasize that it is crucial for \Cref{thm:perfect} and its corollaries that the number of nodes in the topology is larger than the number of agents. 
If the two numbers are equal, then perfect assignments do not exist (unless there is only one type of agents), and the corresponding decision problem becomes trivial. 
Nevertheless, it remains interesting to ask for assignments that are ``individually optimal'' for all agents. 

\begin{definition}
An assignment $\vv$ is called \emph{individually optimal for agent $i$} if $u_i(\vv) \ge u_i(\vv')$ for all assignments $\vv'$. An assignment is called \emph{individually optimal} if it is individually optimal for all agents.
\end{definition}

An example of an individually optimal assignment is shown in \Cref{fig:individually-optimal}, where every agent receives a utility of $1/2$.

\begin{figure}
\centering
\begin{tikzpicture}[scale=0.7]
\draw  (2,0) -- (2,4);
\draw  (2,0) to[out=120,in=240] (2,4);
\draw  (5,0) -- (5,4);
\draw  (5,0) to[out=60,in=-60] (5,4);
\draw  (5,0) -- (2,2) -- (5,4) -- (2,4) -- (5,2) -- (2,0) -- (5,0);
\draw[fill=white] (2,2) circle [radius = 0.3];
\node[blue] at (2,2) {B};
\draw[fill=white] (2,0) circle [radius = 0.3];
\node[blue] at (2,0) {B};
\draw[fill=white] (2,4) circle [radius = 0.3];
\node[blue] at (2,4) {B};
\draw[fill=white] (5,2) circle [radius = 0.3];
\node[red] at (5,2) {R};
\draw[fill=white] (5,0) circle [radius = 0.3];
\node[red] at (5,0) {R};
\draw[fill=white] (5,4) circle [radius = 0.3];
\node[red] at (5,4) {R};
\end{tikzpicture}
\caption{Example of an individually optimal assignment.}
\label{fig:individually-optimal}
\end{figure}
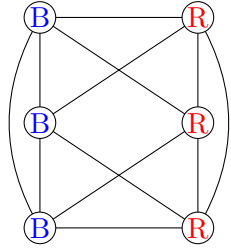

Note that it is easy to compute the utility that an agent receives in an individually optimal assignment for her: Assign this agent to a node of minimum degree; then, assign agents of the same type to as many neighbors as possible and leave other neighbors empty; finally, if needed, assign agents from the other type to the remaining neighbors. 
An individually optimal assignment for all agents, if it exists, is clearly optimal with respect to all welfare measures that we consider and treats agents of the same type equally. 

In the reduction of \Cref{thm:perfect}, we can assume that there are at least three agents of each type, and some nodes in the topology have\footnote{More precisely, we can assume without loss of generality that $|S| \ge 2 + |R|/3$ in an instance $(R,S)$ of X3C, since the problem can be solved by brute force otherwise. Then, there are at least three blue agents. Also, there are at least three red agents whenever $|R|\ge 3$. Finally, nodes of the type $s_{2|S|+2}$ have degree $2$.} degree $2$. Hence, an assignment is individually optimal if and only if every agent receives utility~$1$. Consequently, perfect assignments and individually optimal assignments coincide, and we obtain the following:

\begin{corollary}
Deciding whether there exists an individually optimal assignment is NP-complete.
\end{corollary}

However, we show next that if the numbers of agents and nodes are equal, then this decision problem becomes efficiently solvable. 
More precisely, we provide simple conditions characterizing the instances with an individually optimal assignment. 
Recall the definition of a complement graph from Footnote~\ref{fn:complement}.

\begin{theorem}
\label{thm:individually-optimal}
Given a Schelling instance with an equal number of agents and nodes, there exists an individually optimal assignment if and only if all of the following three conditions are met:
\begin{enumerate}[(1)]
    \item The numbers of red and blue agents are the same, or the topology graph is a complete graph.
    \item The topology graph is regular.
    \item The complement graph of the topology graph is bipartite.
\end{enumerate}
Deciding whether the three conditions are met can be done in polynomial time.
Moreover, if all three conditions are met, we can compute an individually optimal assignment in polynomial time.
\end{theorem}

\begin{proof}
Consider a Schelling instance with an equal number of agents and nodes and a topology graph $G = (V,E)$. 
Denote the complement graph by $\overline G$. 
Note that if $G$ is a complete graph, then all assignments are individually optimal and conditions (2) and (3) are trivially met. 
Therefore, assume from now on that $G$ is not a complete graph. 

First, suppose that there exists an individually optimal assignment, and consider such an assignment.
Recall that there are $b$ blue and $r$ red agents. 
Let $\delta$ denote the minimum degree among the nodes of $G$. 
If $b\ge \delta +1$ or $r\ge \delta +1$, then the individually optimal utility for one type of agents is $1$, and there does not exist an individually optimal assignment.
Hence, this cannot be the case, and the individually optimal utility is $(b-1)/\delta$ and $(r-1)/\delta$ for the blue and red agents, respectively. 
In particular, the topology graph must be $\delta$-regular; otherwise, an agent assigned to a node with degree greater than $\delta$ cannot receive her individually optimal utility. 
Condition (2) is therefore met.

Let $B$ and $R$ be the sets of nodes occupied by the blue and red agents, respectively. 
For the assignment to be individually optimal, $B$ and $R$ must form cliques in the topology graph. 
In other words, $(B,R)$ forms a bipartition of the complement graph, meaning that condition (3) is met.

The regularity of the topology graph implies the regularity of its complement graph, which is not the empty graph (because we already excluded a complete topology graph). 
Since $(B,R)$ forms a bipartition of the complement graph which is $(n-1-\delta)$-regular, the number of edges adjacent to exactly one node in $B$ is $(n-1-\delta)b$, the number of edges adjacent to exactly one node in $R$ is $(n-1-\delta)r$, and these two numbers must be equal.
It follows that $b = r$, so condition (1) is met as well.

Conversely, assume that the three conditions are met. 
Since we have already handled the case of a complete topology graph, we have that the numbers of blue and red agents are the same.

Let $(X,Y)$ be a bipartition of the node set $V$ in $\overline G$. 
Since $G$ is regular, so is $\overline{G}$.
Again, by counting the (nonzero) number of edges incident to nodes in $X$ and $Y$, regularity implies that $|X| = |Y|$. 
Now, consider the assignment that assigns all blue agents at the nodes in $X$ and all red agents at the nodes in $Y$; this assignment is feasible because the numbers of blue and red agents are the same. 
Since $(X,Y)$ forms a bipartition in the complement graph, the set of blue agents as well as the set of red agents form cliques in the topology graph. 
Regularity therefore implies that every agent is individually optimal, so the assignment is individually optimal.

The proof of the converse also shows that we can compute an individually optimal assignment in polynomial time by computing the complement graph and an arbitrary bipartition. 
This yields an individually optimal assignment, provided that all three conditions are met.
It is clear that checking whether the three conditions are met can also be done in polynomial time.
\end{proof}

The key algorithmic problem in the proof of \Cref{thm:individually-optimal} is to decide whether the nodes of a given regular graph can be partitioned into two equally-sized subsets in such a way that each subset forms a clique.
This problem is related to some NP-hard problems, and its tractability may therefore be of broader interest. 
Indeed, it appears similar not only to the \textsc{Maximum Clique} problem, which we have seen to be NP-hard for regular graphs (cf. \Cref{thm:SW-hardness}), but also to the \textsc{Minimum Bisection} problem, which is likewise NP-hard for regular graphs \citep{BuiChLe87}. The latter problem asks for a partition of the nodes of a given graph into two equally-sized subsets such that the number of edges between these two sets is minimized. 
Our proof of \Cref{thm:individually-optimal} shows that the variant where we ask for the two subsets of nodes to form cliques is solvable in polynomial time.

\section{Number of Positive Agents}\label{sec:positive}

In this section, we consider the problem of maximizing the number of agents receiving a positive utility, who we refer to as \emph{positive agents}. 
This problem is closely related to egalitarian and Nash welfare, because an assignment has nonzero egalitarian (resp., Nash) welfare if and only if it makes every agent positive.
Notice that it is not always possible to make every agent positive---for example, in a star, every agent whose type is different from the center agent receives zero utility.
We begin by showing that for trees, deciding whether it is possible to make every agent positive can be done efficiently.
Our algorithm is based on dynamic programming and shares some similarities with the algorithm of \citet{ElkindGaIg19} for deciding whether an equilibrium exists on a tree.

\begin{theorem}
\label{thm:NP-positive-agents}
There is a polynomial-time algorithm that decides whether there exists an assignment in which every agent receives a positive utility when the topology is a tree.
\end{theorem}

\begin{proof}
Pick an arbitrary node $v_{\text{root}}$ to be the root of $G$.
For each node $v\in V$, let $\tree(v)$ be the set of descendants of $v$, including $v$ itself.
For each $v$, we fill out a table $\tau_v$, which contains an entry $\tau_v(C,n_B,n_R,n_E,q)$ for each tuple $(C,n_B,n_R,n_E,q)$, where
\begin{itemize}
\item $C\in\{\text{\em blue}, \text{\em red}, \text{\em empty}\}$;
\item $n_B,n_R,n_E\in\{0,1,\dots,n\}$;
\item $q\in\{\text{\em yes}, \text{\em no}\}$.
\end{itemize}
The number of entries in each table is $6(n+1)^3$.
The value of each entry is either \emph{true} or \emph{false}.
Specifically, $\tau_v(C,n_B,n_R,n_E,q) = \true$ if and only if there exists an assignment of a subset of agents to the nodes in $\tree(v)$ satisfying the following conditions:
\begin{enumerate}
\item If $C=\ept$, then node $v$ is empty; otherwise, it is assigned to an agent of color $C$.
\item There are $n_B$ blue agents, $n_R$ red agents, and $n_E$ empty nodes in $\tree(v)$.
\item If $C\in\{\blue,\red\}$, then $q=yes$ if and only if the agent in node $v$ has at least one child of the same color.
\item Every node in $\tree(v)$ different from $v$ has at least one neighbor of the same color.
\end{enumerate}
An assignment in which every agent receives a positive utility exists if and only if in the table $\tau_{v_{\text{root}}}$ of the root node $v_{\text{root}}$, there exists $(C,n_B,n_R,n_E,q)$ such that $\tau_{v_{\text{root}}}(C,n_B,n_R,n_E,q) = \true$, $n_B = b$, $n_R = r$, and if $C\in\{\blue,\red\}$ then $q = yes$.

The tables for the leaf nodes can be filled in trivially.
We now show how to fill the table $\tau_v$ of each $v\in V$ given the tables of its children.
If $v$ has $L$ children $w_1,\dots,w_L$, we construct intermediate tables $\theta_v^0,\theta_v^1,\dots,\theta_v^L$.
Each table $\theta_v^i$ takes parameters $(n_B,n_R,n_E,q_B,q_R,\widehat{q_B},\widehat{q_R})$.
The entry of the table is set to $\true$ if it is possible to place agents in the first $i$ subtrees so that the following conditions hold: there are a total of $n_B$ blue agents, $n_R$ red agents, and $n_E$ empty nodes; $q_B$ (resp., $q_R$) indicates whether there is at least one blue (resp., red) agent among the first $i$ children of $v$; $\widehat{q_B}$ (resp., $\widehat{q_R}$) indicates whether there is at least one blue (resp., red) agent among the first $i$ children of $v$ who does not have a blue (resp., red) child.
The table $\theta_v^0$ can be filled in trivially: only the entry $\theta_v^0(0,0,0,no,no,no,no)$ is set to \emph{true}.
By combining the tables $\theta_v^{i-1}$ and $\tau_{w_i}$, we can fill in the table $\theta_v^i$ in polynomial time.
Specifically, we set the entry $\theta_v^i(n_B,n_R,n_E,q_B,q_R,\widehat{q_B},\widehat{q_R})$ to \emph{true} if and only if there exist entries $\theta_v^{i-1}(n_B',n_R',n_E',q_B',q_R',\widehat{q_B}',\widehat{q_R}')$ and $\tau_{w_i}(C'',n_B'',n_R'',n_E'',q'')$ such that both are \emph{true} and the following three conditions hold:
\begin{enumerate}
\item $n_B'+n_B''=n_B$, $n_R'+n_R''=n_R$, and $n_E'+n_E''=n_E$.
\item $q_B = yes$ if and only if $q_B' = yes$ or $C'' = blue$. Analogously, $q_R = yes$ if and only if $q_R' = yes$ or $C'' = red$.
\item $\widehat{q_B} = yes$ if and only if (i) $\widehat{q_B}' = yes$ or (ii) $C'' = blue$ and $q''= no$. Analogously, $\widehat{q_R} = yes$ if and only if (i) $\widehat{q_R}' = yes$ or (ii) $C'' = red$ and $q''= no$.
\end{enumerate}
The table $\theta_v^L$ can then be used to fill in $\tau_v$ in polynomial time.
Specifically, we set the entry $\tau_v(C,n_B,n_R,n_E,q)$ to $\true$ if and only if there exists an entry $\theta_v^L(n_B',n_R',n_E',q_B',q_R',\widehat{q_B}',\widehat{q_R}')$ which has been set to $\true$ and such that the following conditions hold:
\begin{enumerate}
\item If $C=\blue$, then $n_B=n_B'+1$, $n_R=n_R'$, and $n_E=n_E'$. Else, if $C=\red$, then $n_B=n_B'$, $n_R=n_R'+1$, and $n_E=n_E'$. Finally, if $C=\ept$, then $n_B=n_B'$, $n_R=n_R'$, and $n_E=n_E'+1$.
\item If $C=\blue$, then $q=q_B'$. Else, if $C=\red$, then $q=q_R'$.
\item If $\widehat{q_B}' = yes$, then $C=\blue$. Similarly, if $\widehat{q_R}' = yes$, then $C=\red$.
\end{enumerate}
This concludes the proof.
\end{proof}

Observe that for any topology, an assignment in which at least half of the agents are positive is guaranteed to exist and can be easily found by using depth-first search for the majority type.

\begin{proposition}
\label{prop:positive-half}
For any $n\ge 3$, there exists a polynomial-time algorithm that computes an assignment in which at least $\lceil n/2\rceil$ agents receive a positive utility.
\end{proposition}

\begin{proof}
Assume without loss of generality that there are at least as many red as blue agents, so there are at least $\lceil n/2\rceil\ge 2$ red agents. 
Starting from an arbitrary node of the topology, we first assign the red agents to nodes as we perform a depth-first search, and then assign the blue agents to any subset of the remaining nodes. 
Since the topology is connected, every red agent will have at least one red neighbor, meaning that at least $\lceil n/2\rceil$ agents receive a positive utility.
\end{proof}

The bound $\lceil n/2\rceil$ is tight when the topology is a star and there are $\lceil n/2\rceil$ red and $\lfloor n/2\rfloor$ blue agents.

Next, we show that when every node has degree at least $2$ and the number of agents is equal to the number of nodes, it is possible to give every agent a positive utility.
Note that the latter condition is also necessary---for the topology given in \Cref{fig:no-positive}, if there are three red and three blue agents (so one node is left unoccupied), it is easy to see that no assignment makes every agent positive.

\begin{figure}
\centering
\begin{tikzpicture}[scale=0.7]
\draw (7,7.73) -- (5,4.27) -- (4,6) -- (8,6) -- (7,4.27) -- (5,7.73) -- (7,7.73);
\draw[fill=white] (6,6) circle [radius = 0.3];
\draw[fill=white] (7,7.73) circle [radius = 0.3];
\draw[fill=white] (5,7.73) circle [radius = 0.3];
\draw[fill=white] (4,6) circle [radius = 0.3];
\draw[fill=white] (5,4.27) circle [radius = 0.3];
\draw[fill=white] (7,4.27) circle [radius = 0.3];
\draw[fill=white] (8,6) circle [radius = 0.3];
\end{tikzpicture}
\caption{Example showing that \Cref{thm:deg-2-positive} does not hold when the number of nodes is greater than the number of agents. There are three red and three blue agents. No matter how the agents are placed, at least one of them will receive utility~$0$.}
\label{fig:no-positive}
\end{figure}
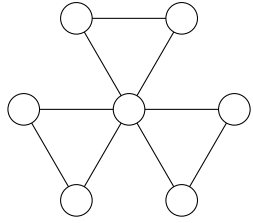

\begin{theorem}
\label{thm:deg-2-positive}
Suppose that every node in the topology has degree at least $2$, the number of agents is equal to the number of nodes, and there are at least two agents of each type.
Then there exists an assignment such that every agent receives a positive utility.
\end{theorem}

\begin{proof}
Consider an arbitrary assignment $\vv$.
If every agent is already positive, we are done, so assume that there is an agent $i$ with utility $0$.
Without loss of generality, $i$ is a blue agent.
Among all paths from $i$ to another blue agent going through only red agents in between, consider one with maximum length---suppose that the path goes to agent $j$.
Since there are at least two blue agents, such a path must exist; moreover, since $i$ has utility $0$, the path contains at least one red agent.
Let $k$ be the last red agent on the path before reaching $j$.
Swap $i$ and $k$.

We claim that in the resulting assignment $\vv'$, the number of positive agents increases by at least $1$; by applying such swaps repeatedly, we will reach an assignment in which all agents are positive.
To establish the claim, it suffices to show that $i$, $k$, as well as any agent adjacent to either of them are positive in $\vv'$.
Since $i$ has utility $0$ in $\vv$, she has at least two red neighbors in $\vv$, so $k$ is positive in $\vv'$.
Moreover, $i$ is adjacent to $j$ in $\vv'$ and therefore becomes positive.
Any other red agent on the path remains positive, and all agents adjacent to $k$ in $\vv'$ are red (besides possibly $i$, if $k$ is the only red agent on the path) and are therefore positive.
Finally, consider any red agent $\ell$ adjacent to $i$ in $\vv'$ not lying on the path.
Since every node has degree at least $2$, agent $\ell$ must have a neighbor $m\neq i$ (possibly $m=j$).
If $m$ is a blue agent (in particular, $m$ is not the same as any node on the path besides $j$), we obtain a longer path from $i$ to $m$ in $\vv$ than the original longest path, a contradiction.
Hence $m$ must be red, and $\ell$ is positive in $\vv'$, proving the claim.
\end{proof}

Since the longest path problem is known to be NP-hard \citep[p.~213]{GareyJo79}, the proof of \Cref{thm:deg-2-positive} does not give rise to a polynomial-time algorithm for computing a desired assignment.
In \Cref{thm:deg-2-positive-algo}, we present an inductive approach that is more involved but leads to an efficient algorithm.

For our final result of this section, we show that maximizing the social welfare and maximizing the number of positive agents can be conflicting goals.
Recall that an assignment has nonzero egalitarian welfare if and only if it makes every agent positive.
To avoid confusion, we will refer to our main notion of social welfare (i.e., the sum of the agents' utilities) as \emph{utilitarian welfare}.

\begin{proposition}
\label{prop:util-egal}
There exists a Schelling instance in which the maximum egalitarian welfare is nonzero but the egalitarian welfare of every assignment that maximizes the utilitarian welfare is zero.
\end{proposition}

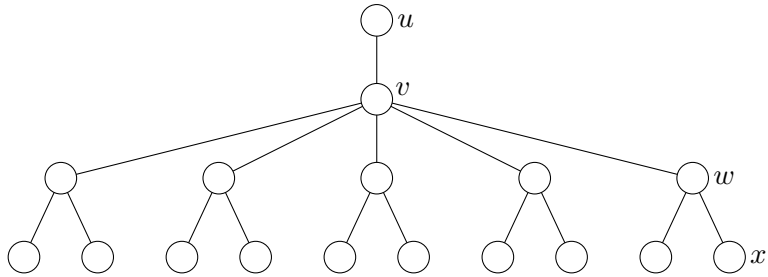
\begin{figure}
\centering
\begin{tikzpicture}[scale=0.7]
\draw (6,9) -- (6,6) -- (6.7,4.5);
\draw (6,6) -- (5.3,4.5);
\draw (3,6) -- (6,7.5) -- (9,6);
\draw (0,6) -- (6,7.5) -- (12,6);
\draw (2.3,4.5) -- (3,6) -- (3.7,4.5);
\draw (8.3,4.5) -- (9,6) -- (9.7,4.5);
\draw (11.3,4.5) -- (12,6) -- (12.7,4.5);
\draw (-0.7,4.5) -- (0,6) -- (0.7,4.5);
\draw[fill=white] (6,9) circle [radius = 0.3];
\draw[fill=white] (6,7.5) circle [radius = 0.3];
\draw[fill=white] (6,6) circle [radius = 0.3];
\draw[fill=white] (6.7,4.5) circle [radius = 0.3];
\draw[fill=white] (5.3,4.5) circle [radius = 0.3];
\draw[fill=white] (9,6) circle [radius = 0.3];
\draw[fill=white] (8.3,4.5) circle [radius = 0.3];
\draw[fill=white] (9.7,4.5) circle [radius = 0.3];
\draw[fill=white] (12,6) circle [radius = 0.3];
\draw[fill=white] (11.3,4.5) circle [radius = 0.3];
\draw[fill=white] (12.7,4.5) circle [radius = 0.3];
\draw[fill=white] (3,6) circle [radius = 0.3];
\draw[fill=white] (3.7,4.5) circle [radius = 0.3];
\draw[fill=white] (2.3,4.5) circle [radius = 0.3];
\draw[fill=white] (0,6) circle [radius = 0.3];
\draw[fill=white] (0.7,4.5) circle [radius = 0.3];
\draw[fill=white] (-0.7,4.5) circle [radius = 0.3];
\node at (6.55,9) {$u$};
\node at (6.5,7.7) {$v$};
\node at (12.6,6) {$w$};
\node at (13.25,4.5) {$x$};
\end{tikzpicture}
\caption{Illustration for the proof of Proposition~\ref{prop:util-egal}.}
\label{fig:util-egal}
\end{figure}

\begin{proof}
Consider the topology in \Cref{fig:util-egal}, and assume that there are $2$ blue and $15$ red agents (so no node can be left empty).
The only way to achieve nonzero egalitarian welfare is to assign the blue agents to $u$ and $v$, yielding utilitarian welfare $14.5$.
However, assigning the blue agents to $w$ and $x$ results in a higher utilitarian welfare of $91/6\approx 15.17$.
\end{proof}

The utilitarian welfare gap in the example in \Cref{prop:util-egal} is rather small. 
However, one can easily modify the example by adding more three-node gadgets as subtrees of the node $v$ to obtain instances wherein the utilitarian welfare of the unique assignment yielding a positive egalitarian welfare is an additive factor of $\Theta(n)$ lower than the maximum possible utilitarian welfare.
Moreover, even assignments that maximize the utilitarian welfare can differ significantly in terms of egalitarian welfare. In Propositions~\ref{prop:egal-zeropos} and \ref{prop:egal-highlow}, we present examples illustrating the possibility that one such assignment has a positive egalitarian welfare while another one has zero, or that two such assignments have positive egalitarian welfare differing by a multiplicative factor of $\Theta(n)$.

\section{Conclusion and Future Work}

In this paper, we have studied questions regarding welfare guarantees and complexity in Schelling segregation.
Several of our findings are positive: An assignment with high social welfare always exists and can be found efficiently, and the welfare of assignments satisfying most optimality notions are at most a linear factor away from the maximum social welfare in the worst case.
Furthermore, even though an assignment yielding a positive utility to every agent may not exist, the existence can be guaranteed when every node in the topology has degree at least $2$, a realistic assumption in well-connected metropolitan areas.
By contrast, computing an assignment that maximizes the social welfare or satisfies any of the optimality notions is NP-hard, and assignments maximizing the (utilitarian) social welfare can differ significantly in terms of egalitarian welfare.

A number of interesting directions remain from our work.
On the technical side, it would be useful to close the gap on the price of Pareto optimality, which we conjecture to be $\Theta(n)$, as well as to characterize the topologies for which an assignment such that every agent receives a positive utility always exists.
Another question is whether we can obtain in polynomial time a better approximation of social welfare than the factor of $2$ in \Cref{thm:derandomized}, or whether there is in fact an inapproximability result.
From a more conceptual perspective, one could try to extend our results to a model with more than two types of agents or more complex friendship relations (e.g., friendship relations defined by a social network, \citeauthor{ElkindGaIg19}, \citeyearR{ElkindGaIg19}) 
or modified utility functions \citep{KanellopoulosKyVo20}.
Questions concerning the convergence behavior in best-response dynamics also remain open: do such dynamics always converge to an optimal assignment?
Finally, studying our new optimality notions from \Cref{sec:notions} in related settings such as hedonic games, especially when agents are partitioned into types, or optimality notions derived from other known welfare measures, may lead to intriguing discoveries as well.

\acks{We would like to thank Jiarui Gan, Pascal Lenzner, and Louise Molitor for interesting discussions, and the anonymous reviewers of AAAI 2021 and JAIR for helpful comments.
}

\vskip 0.2in
\bibliography{main}
\bibliographystyle{theapa}

\appendix
\section{Additional Results}

In the upper bound part of \Cref{thm:uvo-price}, we showed that any UVO assignment has social welfare at least $1/2$. 
Here we establish an improved bound of $1$, albeit with a much longer proof.

\begin{theorem}
\label{thm:UVO-welfare-1}
When $n\ge 3$, any UVO assignment has social welfare at least $1$.
\end{theorem}

\begin{proof}
By an argument similar to that in the upper bound part of \Cref{thm:uvo-price}, it suffices to consider the case where the number of agents is equal to the number of nodes.
Let $\vv$ be a UVO assignment.
If all agents receive nonzero utility, then since the least possible positive utility is $1/(n-1)$, the social welfare would be at least $n/(n-1) > 1$.
Hence we may assume without loss of generality that there is a red agent with utility $0$, and that all blue agents receive a positive utility.

Decompose the topology into maximal monochromatic components.
We claim that there cannot be a set of red components and a set of blue components such that the two sets have the same total number of nodes and together they do not cover all nodes; call this claim (*).
To see why (*) is true, assume for contradiction that two such sets exist.
We swap all agents in the set of red components with those in the set of blue components to obtain an assignment $\vv'$.
The utility of each swapped red agent in $\vv'$ is at least that of the blue agent whom she replaces in $\vv$; an analogous statement holds for the swapped blue agents in $\vv'$.
Moreover, an unswapped agent who is adjacent to at least one swapped agent receives a strictly higher utility in $\vv'$ than in $\vv$.
Hence $\vv$ is not UVO, a contradiction that establishes (*).

Next, assume for contradiction that $\SW(\vv) < 1$, and let $s$ be the number of maximal monochromatic components in the topology.
Call an edge ``monochromatic'' if it connects two agents of the same color.
Since a component with $t$ nodes has at least $t-1$ edges, the total number of monochromatic edges is at least $n-s$.
A monochromatic edge generates a utility of at least $1/(n-1)$ for each of the two agents adjacent to it, so the generated utility is at least $2/(n-1)$ for each edge.
Since $\SW(\vv) < 1$, this means that the number of monochromatic edges is less than $(n-1)/2$, and therefore at most $(n-2)/2$.
Thus, we have $n-s\leq (n-2)/2$, which implies that $n\leq 2s-2$.

Let $k\ge 1$ be the number of red components of size $1$, i.e., the number of red agents with utility $0$.
By (*), the size of any blue component is at least $k$.
Suppose the smallest blue component has size at least $k+1$.
Then, considering the $k$ singleton red components and the smallest blue component, there are $k+1$ components with a total of at least $2k+1$ nodes.
Letting $n_0$ and $s_0$ denote the number of nodes and components considered so far, we have $n_0 > 2s_0-2$.
Since every other component has size at least $2$, we also have $n > 2s-2$, a contradiction with $n\le 2s-2$.
So the smallest blue component must have size $k$.
If there is any other component, we are again done by (*).
Hence, the topology consists exactly of $k$ red components of size $1$ and one blue component of size $k$.
In particular, $k\ge 2$.

We now show that the blue component of size $k$ is a star, i.e., all but one blue agents have exactly one blue neighbor.
Assume that at least two blue agents have more than one blue neighbor each.
Each of these two agents receives utility at least $\frac{2}{k+2}$, while every other blue agent receives at least $\frac{1}{k+1}$.
Hence $\SW(\vv)\ge 2\cdot\frac{2}{k+2} + (k-2)\cdot\frac{1}{k+1} = \frac{k(k+4)}{(k+1)(k+2)}$, which is at least $1$ because $k\ge 2$.
So the blue component is a star, and $\SW(\vv)\ge \frac{k-1}{2k-1} + (k-1)\cdot\frac{1}{k+1} = \frac{3k(k-1)}{(k+1)(2k-1)}$, which is at least $1$ whenever $k\ge 4$.

It remains to consider the cases $k=2$ and $k=3$.
Since $\SW(\vv) < 1$, each blue agent must have at least one red neighbor.
Moreover, between the two blue agents with one blue neighbor, at least one must have more than one red neighbor.
Let $i$ be such a blue agent, $j$ be her blue neighbor, $\ell$ be a red neighbor of $j$, and $m\neq \ell$ be a red neighbor of $i$.
We swap $i$ and $\ell$ to obtain an assignment $\vv'$.
From $\vv$ to $\vv'$, $j$ receives the same utility, $m$ receives a strictly higher utility, $i$ receives a higher utility than $\ell$'s utility in $\vv$ (which is $0$), and $\ell$ receives at least as much utility as $i$ does in $\vv$.
It follows that $\vv$ is not UVO, a final contradiction which completes the proof.
\end{proof}

Next, we present an efficient algorithm for computing an assignment that gives every agent a positive utility when every node has degree at least $2$; a shorter proof that such an assignment exists can be found in \Cref{thm:deg-2-positive}.

\begin{theorem}
\label{thm:deg-2-positive-algo}
Suppose that every node in the topology has degree at least $2$, the number of agents is equal to the number of nodes, and there are at least two agents of each type.
Then, it is possible to compute an assignment in which every agent receives a positive utility in polynomial time.
\end{theorem}

\begin{proof}
We present an inductive approach which inherently gives rise to a polynomial-time algorithm.
First, if there is an edge connecting two nodes of degree at least $3$, deleting it still leaves a topology in which every node has degree at least $2$.
The topology may stay connected or break into two connected components.
We first show how to deal with a connected topology in which no edge connects two nodes of degree at least $3$, and specify later how to proceed when the topology breaks into two components upon the removal of an edge. 
Assume that we have a connected topology such that no edge connects two nodes of degree at least $3$.
If the topology is a cycle, a desired assignment can be easily found, so assume otherwise.

Call the nodes of degree at least $3$ ``primary nodes'', and the remaining nodes (which have degree $2$) ``secondary nodes''.
Note that no edge connects two primary nodes and there exists at least one primary node due to our previous assumptions.
In addition, each secondary node belongs either to a path connecting two primary nodes or to a cycle going from a primary node back to itself.
We prove our claim for the class of graphs satisfying these conditions, but more generally without the assumption that primary nodes have degree at least $3$ (so they may have degree $1$ or $2$).
Nevertheless, the fact that the primary nodes under our original assumptions can be easily identified by their degree will be useful for constructing an efficient algorithm.
From this point on, primary nodes will remain primary throughout the procedure regardless of their degree.

Consider the ``meta-graph'' $H$ whose nodes correspond to the primary nodes of our topology $G$, where two nodes in $H$ are connected by an edge if and only if there exists a path connecting the two corresponding primary nodes not going through any other primary node in $G$.
Since $G$ is connected, so is $H$ (note that $H$ may consist of a single node).
Let $v$ be a primary node in $G$ such that removing the corresponding node in $H$ leaves $H$ connected---for example, any leaf in a spanning tree of $H$ satisfies this property.

We will color certain nodes in $G$ blue in the following order.
Start with a cycle connecting $v$ back to itself (if there is any), and color nodes excluding $v$ in sequence starting from a node adjacent to $v$; then move on to the next cycle if there is another left.
The only exception is if the total amount of nodes to be colored blue only allows us to color one more node when we are about to start coloring a new cycle---in this case, we color $v$ blue instead of the first node of a new cycle.
Whenever we run out of nodes to be colored blue, we color all of the remaining nodes red.
If we have colored all nodes in cycles adjacent to $v$ and there are still nodes left to be colored blue, we color $v$ blue, followed by secondary nodes on paths adjacent to $v$ (call this set of secondary nodes $A$); for each such path, we color nodes closer to $v$ before those further away from $v$.
If we run out of blue nodes, color all remaining nodes red. 
In case we have also colored all nodes in $A$ and removed them, the remaining topology is again connected due to our choice of $v$.
If we have only one node left to be colored blue, color one of the primary nodes adjacent to a secondary node in $A$ blue, and all of the remaining nodes red; any remaining primary node is still adjacent to at least one secondary node.
Else, there are at least two nodes left to be colored blue, and we apply induction on the remaining topology.
Note that if we only have one primary node $v$ left, $G$ is a union of cycles that only intersect each other at $v$, and the same procedure still applies.

We now show how to proceed when, upon removing an edge connecting two nodes of degree at least $3$, the topology breaks into two connected components $C_1$ and $C_2$.
We try to allocate an appropriate number of red and blue agents to each component, and recurse on the two smaller problems.
Assume without loss of generality that $|C_1|\le |C_2|$ and $r\le b$.
Since every node in the remaining (disconnected) topology still has degree at least $2$, we must have $|C_1|\ge 3$.
Assume first that $|C_1|\ge 4$.
If $r\ge 4$, we allocate two red and two blue agents to each component, and the remaining agents arbitrarily so that the number of agents is equal to the number of nodes in each subproblem---this ensures that the subproblems satisfy the conditions of the original problem.
If $r=2$, or if $r=3$ and $|C_2|\ge 5$, we simply allocate all red agents to $C_2$.
The remaining case is that $r=3$, $|C_1|=|C_2| = 4$, and $b=5$.
In this case, suppose that the removed edge is adjacent to node $x$ in $C_1$. 
We assign the blue agents to $C_2\cup\{x\}$, and the red agents to $C_1\setminus\{x\}$.
Since each node in $C_1$ is adjacent to at least one other node in $C_1$ besides $x$, all agents are positive.

Consider now the case where $|C_1| = 3$.
If one of the types has three or at least five agents, we may allocate three agents of that type to $C_1$.
Hence, assume that both types consist of either two or four agents.
The only two possibilities are then $(|C_1|,|C_2|,r,b) = (3,3,2,4)$ and $(3,5,4,4)$.
Both cases can be handled similarly to the case $(4,4,3,5)$ in the previous paragraph.

Since each step of our procedure removes at least one node or edge, by following the procedure, we obtain a polynomial-time algorithm that computes a desired assignment.
\end{proof}

In fact, the algorithm in \Cref{thm:deg-2-positive-algo} runs in time linear in the size of the input.
Indeed, finding a spanning tree of $H$ can be done by breadth-first search, and all other steps of the algorithm take time $O(m)$, where $m$ denotes the number of edges in $G$.

Finally, we continue our discussion after Proposition~\ref{prop:util-egal} by presenting further examples relating egalitarian and utilitarian welfare.

\begin{proposition}
\label{prop:egal-zeropos}
There exists a Schelling instance in which one assignment maximizing utilitarian welfare also maximizes egalitarian welfare, while another such assignment has zero egalitarian welfare.
\end{proposition}

\begin{figure}[!h]
\centering
\begin{tikzpicture}[scale=0.7]
\draw (5.3,9) -- (6,7.5);
\draw (6.7,9) -- (6,7.5);
\draw (4,6) -- (6,7.5) -- (8,6);
\draw (0,6) -- (6,7.5) -- (12,6);
\draw (3,4.5) -- (4,6) -- (4,4.5);
\draw (4,6) -- (5,4.5);
\draw (7,4.5) -- (8,6) -- (9,4.5);
\draw (8,6) -- (8,4.5);
\draw (11,4.5) -- (12,6) -- (13,4.5);
\draw (12,6) -- (12,4.5);
\draw (-1,4.5) -- (0,6) -- (1,4.5);
\draw (0,6) -- (0,4.5);
\draw[fill=white] (6.7,9) circle [radius = 0.3];
\draw[fill=white] (5.3,9) circle [radius = 0.3];
\draw[fill=white] (6,7.5) circle [radius = 0.3];
\draw[fill=white] (8,6) circle [radius = 0.3];
\draw[fill=white] (7,4.5) circle [radius = 0.3];
\draw[fill=white] (8,4.5) circle [radius = 0.3];
\draw[fill=white] (9,4.5) circle [radius = 0.3];
\draw[fill=white] (12,6) circle [radius = 0.3];
\draw[fill=white] (11,4.5) circle [radius = 0.3];
\draw[fill=white] (12,4.5) circle [radius = 0.3];
\draw[fill=white] (13,4.5) circle [radius = 0.3];
\draw[fill=white] (4,6) circle [radius = 0.3];
\draw[fill=white] (4,4.5) circle [radius = 0.3];
\draw[fill=white] (3,4.5) circle [radius = 0.3];
\draw[fill=white] (5,4.5) circle [radius = 0.3];
\draw[fill=white] (0,6) circle [radius = 0.3];
\draw[fill=white] (1,4.5) circle [radius = 0.3];
\draw[fill=white] (0,4.5) circle [radius = 0.3];
\draw[fill=white] (-1,4.5) circle [radius = 0.3];
\node at (4.65,9) {$u_1$};
\node at (7.35,9) {$u_2$};
\node at (6.5,7.7) {$v$};
\node at (12.6,6) {$w$};
\node at (13.7,4.5) {$x_2$};
\node at (10.3,4.5) {$x_1$};
\end{tikzpicture}
\caption{Illustration for the proof of Proposition~\ref{prop:egal-zeropos}.}
\label{fig:egal-zeropos}
\end{figure}
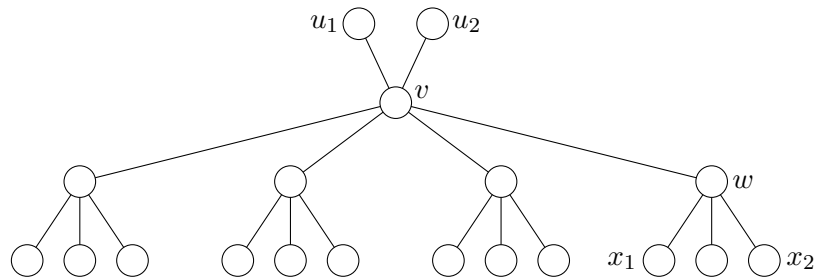

\begin{proof}
Consider the topology in \Cref{fig:egal-zeropos}, and assume there are $3$ blue and $16$ red agents (so no node can be left empty).
The only way to achieve nonzero egalitarian welfare is to assign the blue agents to $u_1$, $u_2$, and $v$, yielding utilitarian welfare $52/3$.
Assigning the blue agents to $w$, $x_1$, and $x_2$ results in the same utilitarian welfare, while the egalitarian welfare is then zero.
It can be verified that both assignments maximize the utilitarian welfare.
\end{proof}

\begin{proposition}
\label{prop:egal-highlow}
There exists a class of Schelling instances such that the ratio between the maximum and minimum egalitarian welfare among assignments that maximize utilitarian welfare is $\Theta(n)$. 
\end{proposition}

\begin{proof}
Let $k$ be a positive integer. 
We define the topology $G = (V,E)$ of a Schelling instance with node set $V$ given by
\[V = \{a_i\colon 1 \le i \le k\}\cup \{x_1,x_2,y_1,y_2\}\cup \{b_i^j\colon 1\le j \le 3, 1\le i \le k\}\] 
and edge set $E$ given by
\begin{itemize}
    \item $\{a_i,a_\ell\}\in E$ for $1\le i,\ell\le k$;
    \item $\{a_i, x_\ell\}\in E$ for $1\le i\le k$, $1\le \ell\le 2$;
    \item $\{x_1,y_1\},\{x_2,y_2\}\in E$;
    \item $\{y_j, b_i^j\}\in E$ for $1\le i\le k$, $1\le j \le 2$;
    \item $\{b_i^j,b_\ell^j\}\in E$ for $1\le i,\ell\le k$, $1\le j\le 3$;
    \item $\{b_i^j,b_i^3\}\in E$ for $1\le i\le k$, $1\le j\le 2$;
    \item no further edges are in $E$.
\end{itemize}
See \Cref{fig:egal-highlow} for an illustration for $k=3$.

\begin{figure}[!h]
\centering
\begin{tikzpicture}[scale=0.7]
\draw (2,8) -- (0,5) -- (2,2) -- (2,8);
\draw (4.5,6.4) -- (0,5) -- (4.5,3.5);
\draw (4.5,6.5) -- (2,8) -- (4.5,3.5);
\draw (4.5,6.5) -- (2,2) -- (4.5,3.5);
\draw (4.5,6.5) -- (7.5,6.5);
\draw (4.5,3.5) -- (7.5,3.5);
\draw (7.5,6.5) -- (10,8);
\draw (7.5,6.5) -- (10,2);
\draw (7.5,6.5) -- (12,5);
\draw (7.5,3.5) -- (14,7);
\draw (7.5,3.5) -- (14,1);
\draw (7.5,3.5) -- (16,4);
\draw (10,8) -- (10,2) -- (12,5) -- (10,8);
\draw (14,7) -- (14,1) -- (16,4) -- (14,7);
\draw (18,8) -- (18,2) -- (20,5) -- (18,8);
\draw (10,8) -- (18,8) -- (14,7);
\draw (12,5) -- (20,5) -- (16,4);
\draw (10,2) -- (18,2) -- (14,1);
\draw[fill=white] (2,8) circle [radius = 0.3];
\draw[fill=white] (2,2) circle [radius = 0.3];
\draw[fill=white] (0,5) circle [radius = 0.3];
\draw[fill=white] (4.5,6.5) circle [radius = 0.3];
\draw[fill=white] (4.5,3.5) circle [radius = 0.3];
\draw[fill=white] (7.5,6.5) circle [radius = 0.3];
\draw[fill=white] (7.5,3.5) circle [radius = 0.3];
\draw[fill=white] (10,8) circle [radius = 0.3];
\draw[fill=white] (10,2) circle [radius = 0.3];
\draw[fill=white] (12,5) circle [radius = 0.3];
\draw[fill=white] (14,7) circle [radius = 0.3];
\draw[fill=white] (14,1) circle [radius = 0.3];
\draw[fill=white] (16,4) circle [radius = 0.3];
\draw[fill=white] (18,8) circle [radius = 0.3];
\draw[fill=white] (18,2) circle [radius = 0.3];
\draw[fill=white] (20,5) circle [radius = 0.3];
\node at (1.3,8) {$a_1$};
\node at (-0.7,5) {$a_2$};
\node at (1.3,2) {$a_3$};
\node at (4.5,7.1) {$x_1$};
\node at (4.5,2.9) {$x_2$};
\node at (7.5,7.1) {$y_1$};
\node at (7.5,2.9) {$y_2$};
\node at (10,8.7) {$b_1^1$};
\node at (12.4,4.5) {$b_2^1$};
\node at (10,1.3) {$b_3^1$};
\node at (13.5,7.5) {$b_1^2$};
\node at (16.4,3.5) {$b_2^2$};
\node at (14,0.3) {$b_3^2$};
\node at (18.7,8) {$b_1^3$};
\node at (20.7,5) {$b_2^3$};
\node at (18.7,2) {$b_3^3$};
\end{tikzpicture}
\caption{Illustration for the proof of Proposition~\ref{prop:egal-highlow}.}
\label{fig:egal-highlow}
\end{figure}
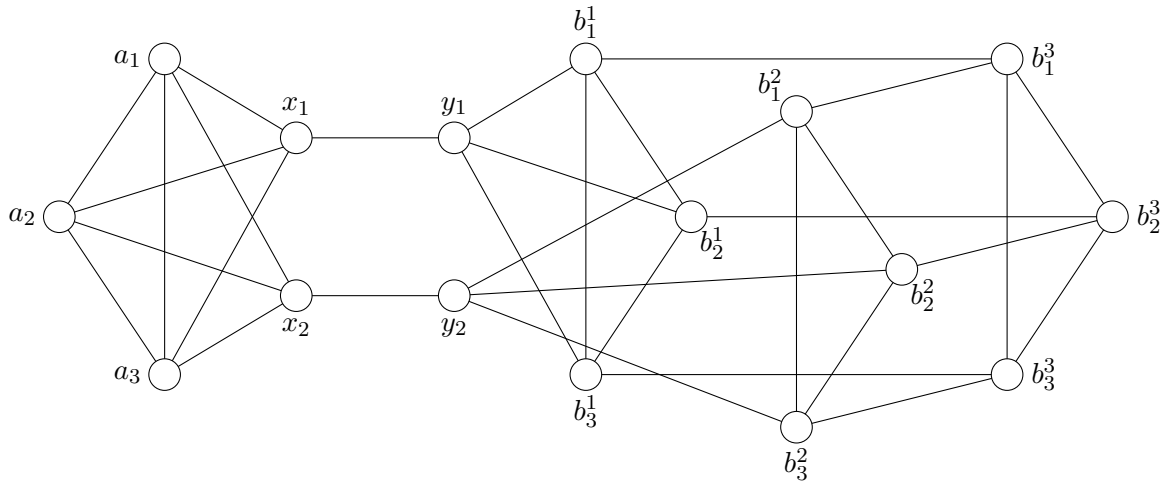

Assume that there are $n = 4k+4$ agents, composed of $k$ blue and $3k + 4$ red agents, so no node can be left empty. 
Note that the topology graph is $(k + 1)$-regular and contains cliques of size $k$. 
Hence, as in the proof of \Cref{thm:SW-hardness}, an assignment maximizes the utilitarian welfare if and only if the blue agents form a clique.

Consider the assignment where the blue agents form the clique $\{a_i\colon 1\le i \le k\}$.
The utility of each neighboring (red) agent $x_j$ is $1/(k+1)$, so the egalitarian welfare is $1/(k+1)$. 
On the other hand, consider the assignment with the blue agents forming the clique $\{b_i^3\colon 1\le i \le k\}$. 
In this assignment, every agent has utility at least $(k-1)/(k+1)$, where the minimum utility is attained by the blue agents. 
Hence, the egalitarian welfare is $(k-1)/(k+1)$, which is a multiplicative factor of $\Theta(k) = \Theta(n)$ higher than that of the first assignment.
\end{proof}

\end{document}